\documentclass[a4paper,USenglish,cleveref, thm-restate]{lipics-v2021}

\hideLIPIcs 
\title{Two Choices are Enough for P-LCPs, USOs, and Colorful Tangents}

\author{Michaela Borzechowski}{Department of Mathematics and Computer Science, Freie Universität Berlin, Germany}{michaela.borzechowski@fu-berlin.de}{}{DFG within GRK~2434 \emph{Facets of Complexity}.}

\author{John Fearnley}{Department of Computer Science, University of Liverpool, UK}{john.fearnley@liverpool.ac.uk}{https://orcid.org/0000-0003-0791-4342}{EPSRC Grant EP/W014750/1.}

\author{Spencer Gordon}{Department of Computer Science, University of Liverpool, UK}{spencer.gordon@liverpool.ac.uk}{}{EPSRC Grant EP/W014750/1.}

\author{Rahul Savani}{The Alan Turing Institute; and Department of Computer Science, University of Liverpool, UK}{rahul.savani@liverpool.ac.uk}{https://orcid.org/0000-0003-1262-7831}{EPSRC Grant EP/W014750/1.}

\author{Patrick Schnider}{Department of Computer Science, ETH Zürich, Switzerland}{patrick.schnider@inf.ethz.ch}{https://orcid.org/0000-0002-2172-9285}{}

\author{Simon Weber}{Department of Computer Science, ETH Zürich, Switzerland}{simon.weber@inf.ethz.ch}{https://orcid.org/0000-0003-1901-3621}{Swiss National Science Foundation under project no. 204320.}

\authorrunning{M. Borzechowski, J. Fearnley, S. Gordon, R. Savani, P. Schnider, S. Weber} 

\Copyright{Michaela Borzechowski, John Fearnley, Spencer Gordon, Rahul Savani, Patrick Schnider, Simon~Weber}



\acknowledgements{We wish to thank Bernd Gärtner for introducing the authors to each other.}

\nolinenumbers 

\usepackage{xspace}
\usepackage{amsmath,amssymb,amsfonts,mathrsfs,amsthm}
\usepackage[capitalise]{cleveref}
\usepackage{tikz}
\usetikzlibrary{decorations.markings}
\usetikzlibrary{shapes}
\usepackage{enumitem}

\newcommand{\abs}[1]{\lvert #1 \rvert}

\newcommand{\f}[1]{\relax\ifmmode#1\else{$#1$}\fi}

\newcommand{\dimension}{\f{d}\xspace}
\newcommand{\Reals}{\mathbb{R}\xspace}
\newcommand{\Naturals}{\mathbb{N}\xspace}
\newcommand{\R}{\Reals}

\newcommand{\xor}{\f{\oplus}\xspace}
\newcommand{\dimensionA}{\f{d}\xspace}
\newcommand{\dimensionB}{\f{n}\xspace}
\newcommand{\simplex}{\f{\Delta}\xspace}

\newcommand{\orientationCube}{\f{O}\xspace}
\newcommand{\orientationGrid}{\f{\sigma}\xspace}

\newcommand{\szabo}{Szab{\'o}}
\newcommand{\SWC}{\szabo-Welzl condition\xspace}

\newcommand{\Comp}[1]{\textsf{#1}\xspace}
\newcommand{\Problem}[1]{\textsc{#1}\xspace}

\usepackage{todonotes}

\newcommand{\lcp}{\Problem{LCP}}
\newcommand{\glcp}{\Problem{GLCP}}
\newcommand{\plcp}{\Problem{P\nobreakdashes-LCP}}
\newcommand{\pglcp}{\Problem{P\nobreakdashes-GLCP}}

\newcommand{\gfixp}{\Problem{Lin-Bellman}}
\newcommand{\pfixp}{\Problem{P\nobreakdashes-Lin-Bellman}}
\newcommand{\swsHS}{\Problem{SWS-Colorful-Tangent}}
\newcommand{\swsTwoHS}{\Problem{SWS-2P-Colorful-Tangent}}
\newcommand{\aHS}{\Problem{$\alpha$\nobreakdashes-Ham-Sandwich}}

\newcommand{\deftab}{0.6cm}
\newcommand{\reals}{\mathbb R}

\DeclareMathOperator{\aff}{aff}
\let\dim\undefined
\DeclareMathOperator{\dim}{dim}

\newcommand{\CubeUSO}{\Problem{Cube-USO}}
\newcommand{\GridUSO}{\Problem{Grid-USO}}

\newcommand{\Grid}{\f{\Gamma}\xspace}

\newcommand{\partitionLength}{\f{n}\xspace}

\newcommand{\vertex}[1]{\f{(#1)}}

\newcommand{\myparagraph}[1]{\medskip \noindent \textbf{#1}}

\begin{document}

\maketitle

\begin{abstract}
We provide polynomial-time reductions between three search problems
from three distinct areas: the P-matrix linear complementarity problem~(P-LCP), finding
the sink of a unique sink orientation~(USO), and a variant of the
$\alpha$-Ham Sandwich problem. For all three settings, we show that ``two
choices are enough'', meaning that the general non-binary version of the problem can be
reduced in polynomial time to the binary version. This specifically means that generalized P-LCPs
are equivalent to P-LCPs, and grid USOs are equivalent to cube USOs.
These results are obtained by showing that both the P-LCP and our $\alpha$-Ham Sandwich variant are equivalent to a new problem we introduce, \pfixp. This problem can be seen as a new tool for formulating problems as P-LCPs.
\end{abstract}

\newpage
\section{Introduction}

\begin{figure}[htb]
    \centering
    \pgfdeclarelayer{bg}    
    \pgfsetlayers{bg,main}  

\begin{tikzpicture}
    \begin{pgfonlayer}{bg}
    \fill[rounded corners, lightgray, opacity=0.5] (1, -1.5) rectangle (3, 4.5) {};
    \fill[rounded corners, lightgray, opacity=0.5] (-7.75, -1.5) rectangle (-2.25, 4.5) {};
    \fill[rounded corners, lightgray, opacity=0.5] (3.25, -1.5) rectangle (5.75, 4.5) {};
    \end{pgfonlayer}
    
    \node (PFixP) at (-0.625,2) {\pfixp};

    \node at (2,4) {\small Algebra};
    \node (PGLCP) at (2,1) {\pglcp};
    \node (PLCP) at (2,3) {\plcp};

    \node at (-5,4) {\small Geometry};
    \node (SWSaHS) at (-5,-1) {\aHS};
    \node (SWS-Lower-Tangent) at (-5,1) {\swsHS};
    \node (SWS-2P-Lower-Tangent) at (-5,3) {\swsTwoHS};

    \node at (4.5,4) {\small Combinatorics};
    \node (aGridUSO) at (4.5,-1) {$\alpha$-\GridUSO};
    \node (GridUSO) at (4.5,1) {\GridUSO};
    \node (CubeUSO) at (4.5,3) {\CubeUSO};

    \path[->, very thick] (PLCP) edge (PGLCP); 
    \path[->, very thick] (CubeUSO) edge[] (GridUSO); 
    \path[->, very thick] (GridUSO) edge (aGridUSO); 
    \path[->, very thick] (PLCP) edge node[midway,below]{\cite{stickney1978digraph}} (CubeUSO); 
    \path[->, very thick] (PGLCP) edge node[midway,above]{\cite{gaertner2008grids}} (GridUSO);
    \path[->, very thick] (SWS-2P-Lower-Tangent) edge (SWS-Lower-Tangent);
    \path[->, very thick] (SWS-Lower-Tangent) edge (SWSaHS);
    
    \path[->, very thick, purple] (GridUSO) edge[bend right] node[sloped, midway, below]{ Thm. \ref{thm:generalizationIsUSO}} (CubeUSO); 

    \path[->, very thick, purple] (SWSaHS) edge node[pos=0.465, below]{\Cref{rem:alphaHam2alphaGrid}
    } (aGridUSO); 

    \path[<->, very thick, purple] (PLCP) edge node[sloped,midway,above]{Thm. \ref{thm:PFixP=PLCP}} (PFixP); 

    \path[->, very thick, purple] (PGLCP) edge[bend left=10] node[sloped,midway, below]{Lem. \ref{lem:PGLCP2PFixP}} (PFixP); 

    \path[->, very thick, purple] (PFixP) edge node[sloped, midway, below]{\Cref{thm:PFixP-to-SWS2PaHS}} (SWS-2P-Lower-Tangent); 

    \path[->, very thick, purple] (SWS-2P-Lower-Tangent) edge[bend left=15] node[midway, above]{\Cref{lem:swsTwoHs2cubeUSO}} (CubeUSO); 

    \path[->, very thick, purple] (SWS-Lower-Tangent) edge[bend right=10] node[sloped,midway, below]{Thm. \ref{thm:SWSaHS2PFixP}} (PFixP); 
    
    \path[->, very thick, purple] (SWS-Lower-Tangent) edge[bend right=20] node[midway, below]{\Cref{lem:swsTwoHs2cubeUSO}} (GridUSO); 

    \end{tikzpicture}
    \caption{Red: Reductions we show in this paper. Black: Existing reductions and trivial inclusions.}
    \label{fig:overview}
\end{figure}
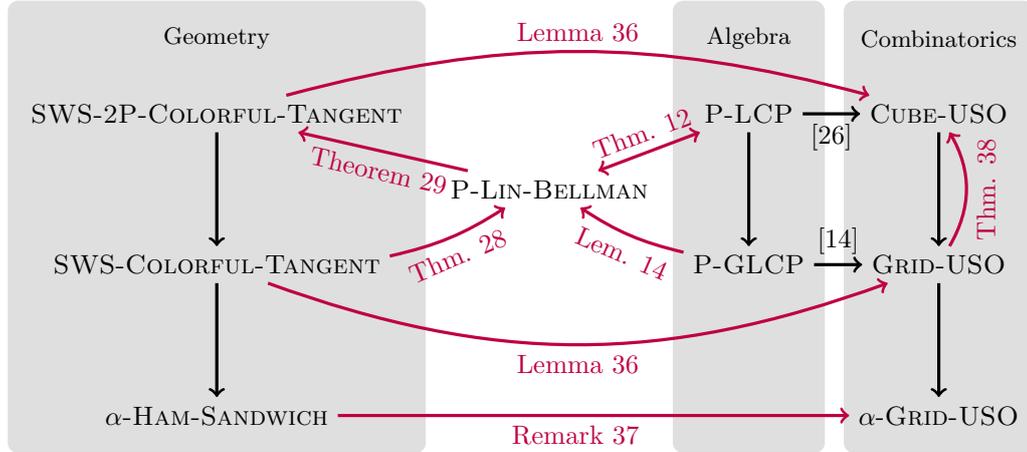

In this paper we study three search problems from three distinct areas: the problem of
solving a P-matrix linear complementarity problem (P-LCP), an algebraic problem, the
problem of finding the sink of a unique sink orientation (USO), a combinatorial
problem, and a variant of the $\alpha$-Ham Sandwich problem, a geometric
problem. Our results are a suite of reductions between these problems, which
are shown in \cref{fig:overview}. 

There are two main themes for these reductions.
\begin{itemize}
\item \textbf{Two choices are enough.} For all three settings, the problems can
be restricted to variants in which all choices are binary. In all three
cases we show that the general non-binary problem can be reduced in
polynomial time to the binary version. 

\item \textbf{A new tool for working with P-LCPs.} We introduce the \pfixp
problem, which serves as a crucial intermediate problem for our reductions with
P-LCPs, and provides a new tool for showing equivalence to the P-LCP problem.
\end{itemize}
We now describe each of the three problems and our results.

\myparagraph{P-Matrix LCPs.}
In the \emph{Linear Complementarity Problem} (LCP), we are given an $n \times n$ 
matrix $M$ and an $n$-dimensional vector $q$, and we are asked to find two
$n$-dimensional vectors $w, z$ such
that $w = Mz + q$, both $w$ and $z$ are non-negative, and $w$ and $z$ satisfy
the \emph{complementarity} condition of $w^Tz = 0$. In general, there is no
guarantee that an LCP has a solution. However, if $M$ is promised to be a P-matrix, that is, 
all its principal minors are positive, then
the LCP problem always has a unique solution for every possible
$q$~\cite{murty1972pmatrix}, and we call this the \plcp problem.

The \plcp problem is important because many optimization problems reduce to it,
for example, \emph{Linear Programming}~\cite{gaertner2006lpuso,HI13} and \emph{Strictly Convex Quadratic Programming}~\cite{CPS09,murty1988linear}, and solving a
\emph{Simple Stochastic Game (SSG)}~\cite{gaertner2005stochasticgames,SvenssonV06}. However, the
complexity status of the \plcp remains a major open question. The \plcp problem is not
known to be polynomial-time solvable, but \Comp{NP}-hardness (in the sense that an oracle for it could be used to solve
\Problem{SAT} in polynomial time) would imply
$\Comp{NP}=\Comp{co-NP}$~\cite{megiddo1988note}.

The \plcp problem naturally encodes problems that have \emph{two choices}. This
property arises from the complementarity condition, where for each index $i$,
one must choose either $w_i = 0$ or $z_i = 0$. 
Thus the problem can be seen as making one of two choices for each of the $n$ possible dimensions of $w$ and $z$. 
For example, this means that the direct encoding of a Simple Stochastic Game as a
\plcp~\cite{HI13} only works for \emph{binary} games, in which each vertex has
exactly two outgoing edges, and for each vertex the choice between the two
outgoing edges is encoded by choosing between $w_i$ and $z_i$.  

To directly encode a non-binary SSG one must instead turn to \emph{generalized}
LCPs, which allow for more than two choices in the complementarity condition~\cite{SvenssonV06}.
Generalized LCPs, which were defined by Habetler and Szanc~\cite{PGLCPUniqueness}, 
also have a P-matrix version, which we will refer to as \pglcp.

\myparagraph{Two choices are enough for P-LCPs.}
Our first main result is to show that \pglcp and \plcp are polynomial-time
equivalent problems. Every \plcp is a \pglcp by definition, so our contribution
is to give a polynomial-time reduction from \pglcp to \plcp, meaning that any
problem that can be formulated as a \pglcp can also be efficiently recast as a \plcp. Such a result was already claimed in 1995, but later a counterexample to a crucial step in the proof was found by the same author~\cite{pglcpProofClaim,pglcpProofMistake}.

To show this result we draw inspiration from the world of infinite games. SSGs
are a special case of \emph{stochastic discounted games} (SDGs), which are
known to be reducible in polynomial-time to the \plcp problem~\cite{HI13,JS08}.
This reduction first writes down a system of \emph{Bellman equations} for the
game, which are a system of linear equations using $\min$ and $\max$
operations. These equations are then formulated as a \plcp using the
complementarity condition to encode the $\min$ and $\max$ operations. 

We introduce a generalization of the Bellman equations for SDGs, which we call
\pfixp. We show that the existing reduction from SDGs to \plcp continues to
work for \pfixp, and we show that we can polynomial-time reduce \plcp to \pfixp. Thus, we
obtain a Bellman-equation type problem that is equivalent to \plcp.

Then we use \pfixp as a tool to reduce \pglcp to \plcp. Specifically we show
that a \pglcp can be reduced to \pfixp. Here our use of \pfixp as an
intermediary shows its usefulness: while it is not at all clear how one can
reduce \pglcp to \plcp directly, when both problems are formulated as \pfixp
instances,
their equivalence essentially boils down to the fact that $\max(a, b, c) =
\max(a, \max(b, c))$. 

\myparagraph{Unique Sink Orientations.}
A Unique Sink Orientation (USO) is an orientation of the $n$-dimensional
hypercube such that every subcube contains a unique sink. An example of a USO can be found in \Cref{fig:USO}. The goal is to find
the unique sink of the entire hypercube. 
USOs were introduced
by Stickney and Watson~\cite{stickney1978digraph} as a combinatorial
abstraction of the candidate solutions of the \plcp and have been studied ever
since \szabo{} and Welzl formally defined them in 2001~\cite{szabo2001usos}. 

\begin{figure}[hbt]
    \centering
    \includegraphics[scale=0.8]{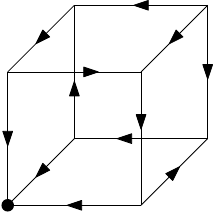}
    \caption{A 3-dimensional USO. The marked vertex denotes its unique global sink.}
    \label{fig:USO}
\end{figure}

USOs have received much attention because many problems are known to reduce to
them. This includes linear programming and more generally convex quadratic
programming, simple stochastic games, and finding the smallest enclosing ball
of given points. 

As with P-LCPs, USOs naturally capture problems in which there are two choices.
Each dimension of the cube allows us to pick between one of two faces, and thus
each vertex of the cube is the result of making $n$ binary choices. To remove this
restriction, prior work has studied unique sink orientations of products of
complete graphs~\cite{gaertner2008grids}, which the literature somewhat
confusingly calls \emph{grid USOs}. For example, as with \plcp, the direct reduction
from SSGs to USOs only works for binary games, while a direct reduction from
non-binary games yields a grid USO.

\myparagraph{Two choices are enough for USOs.}
Our second main result is to show that grid USOs and cube USOs are
polynomial-time equivalent problems. Every cube USO is a grid USO by
definition, so our contribution is to provide a polynomial-time reduction from
grid USOs to cube USOs. 
Despite many researchers suspecting that grid USOs are a strict generalization of cube USOs, our result shows that at least in the promise setting, the two problems are computationally equivalent. For the reduction we embed a $k$-regular grid into a (also $k$-regular) $k$-cube. The main challenge to overcome is that the $k$-cube contains many more vertices and significantly more edges. These edges need to be oriented in such a way to be consistent with the orientation of the grid, and this orientation also needs to be computable locally, i.e., without looking at the whole orientation of the grid.

It should be noted that while \plcp is known to reduce to cube USOs, and
\pglcp is known to reduce to grid USOs, neither of our ``two choices is enough'' results imply each
other. This is because there is no known reduction from a USO-type problem to an LCP-type
problem.

\myparagraph{Equivalence between P-LCP and Colorful Tangent problems.}
The Ham Sandwich theorem is a classical theorem in computational geometry:
given $d$ sets of points in $\R^d$, we can find a hyperplane that
simultaneously \emph{bisects} all of the sets. Steiger and
Zhao~\cite{alphaHamSandwich} have shown that in a restricted input setting we can not only bisect, but cut off arbitrary fractions of each point set. This is the so-called $\alpha$-Ham Sandwich
theorem: if the point sets are \emph{well-separated}, then for each vector 
$\alpha$ there exists a unique choice of one point per set, such that the hyperplane spanned by these points has exactly $\alpha_i$ of the points from set $i$ above. The $\alpha$-Ham Sandwich problem asks to determine this unique choice of one point per set.

In this paper, we consider a restricted variant of the $\alpha$-Ham Sandwich
problem. Firstly, we slightly strengthen the assumption of well-separation to
\emph{strong} well-separation. Secondly we restrict the input vector $\alpha$ such that each entry takes either be the minimum or maximum possible value. This means 
that for every point set, either all points of the set must lie above, or all
below the desired $\alpha$-cut. Thus, the $\alpha$-cuts we search for are
tangents to all the point sets. Combining the assumption of strong
well-separation and the solutions being tangents, we call this problem \swsHS.
We also consider the binary variant of the problem, which we call \swsTwoHS,
where we restrict the size of each point set to $2$; finding an $\alpha$-cut then corresponds to a series of binary choices, one per set.

Here our contribution is to show that \swsHS is polynomial-time equivalent to the \plcp
problem. Our new intermediate problem \pfixp plays a crucial role in this
result: we give a polynomial-time reduction from \swsHS to \pfixp, and a
polynomial-time reduction from \pfixp to \swsTwoHS.
This also shows
that two choices are enough for this problem as well. 

For these reductions, we consider the point-hyperplane dual of the colorful tangent problems.
In the dual, every input point becomes a hyperplane, and the solution we search for is a point lying on one hyperplane of each color, and lying either above or below all other hyperplanes. This can be encoded in the $\min$ and $\max$ operations of the \pfixp problem.
This demonstrates the usefulness of \pfixp as an intermediate problem, since it has
allowed us to show what is -- to the best of our knowledge -- the first polynomial-time 
equivalence between P-LCP and another problem. 
    
\myparagraph{Related work.}
All problems we consider in this paper are promise search problems which lie in the complexity class \Comp{PromiseUEOPL}~\cite{borzechowski2022unique,chiu_et_al-ComplexityAlphaHamSandwich}, which is the promise setting analogue of the class \Comp{UEOPL} (short for Unique End of Potential Line)~\cite{fearnley2020ueopl}. The latter has recently attracted the attention of many researchers trying to further the understanding of the hierarchy of total search problem classes, with a breakthrough result separating \Comp{UEOPL} from \Comp{EOPL}~\cite{separations}. There is only one known natural\footnote{By natural we mean a problem that is not a variant of the Unique End of Potential Line problem which naturally characterizes the class.} complete problem for \Comp{UEOPL}, called One Permutation Discrete Contraction (\Problem{OPDC}). \Problem{OPDC} also admits a natural binary variant, and it has already been shown implicitly that the binary variant is polynomial-time equivalent to the general variant~\cite{fearnley2020ueopl}. Our reductions may help in finding another natural \Comp{UEOPL}-complete problem, even though our results for now only hold in the promise setting.

\myparagraph{Paper overview.} We introduce our new \pfixp problem in \Cref{sec:pfixp}. Then, we show the equivalence of \plcp and \pfixp as well as the reduction from \pglcp to \plcp in \Cref{sec:plcp}. In \Cref{sec:alphaHam} we show the reductions between the colorful tangent problems and \pfixp. Finally, we show the reductions from the \plcp and colorful tangent problems to \CubeUSO and \GridUSO, as well as the reduction from \GridUSO to \CubeUSO in \Cref{sec:gridtocube}.

\section{A New Intermediate Problem: \texorpdfstring{\pfixp}{P-Lin-Bellman}}\label{sec:pfixp}


Our new \pfixp problem is motivated by 
discounted games. A stochastic discounted game (SDG) is defined by a set
of states $S$, which are partitioned into $S_\text{Max}$ and $S_\text{Min}$, a set of actions $A$, a transition function $p : S \times A
\times S \rightarrow \reals$, a reward function $r : S\times A \rightarrow \reals$, and
a discount factor $\lambda$. 
The value of a SDG is known to be the solution to the following system of
\emph{Bellman equations}. For each state $s$ we have an equation
\begin{align*}
x_s &= \begin{cases}
\max_{a \in A} \left( r(s,a) + \lambda \cdot \sum_{s' \in S} p(s, a, s') \cdot x_{s'} \right) &
\text{if $s \in S_{\text{Max}},$} \\
\min_{a \in A} \left( r(s,a) + \lambda \cdot \sum_{s' \in S} p(s, a, s') \cdot x_{s'} \right) &
\text{if $s \in S_{\text{Min}}.$} 
\end{cases}
\end{align*}
Prior work has shown that, if the game is binary, meaning that $|A| = 2$, then
these Bellman equations can be formulated as a \plcp~\cite{JS08,HI13}.

For our purposes, we are interested in the format of these equations. Note that
each equation is a max or min over affine functions of the other variables.
We capture this idea in the following generalized definition.

\begin{definition}
A $\gfixp$ system $G = (L, R, q, S)$ is defined by 
two matrices $L, R \in \reals^{n \times n}$, a vector $q \in \reals^n$, and a
set $S \subseteq \{1, 2, \dots, n\}$. 
These inputs define the following system of equations over a vector of variables $x \in
\reals^n$: 
\begin{equation}
\label{eqn:maxminlin}
x_i = \begin{cases} 
\max( \sum_{1 \le j \le n} L_{ij} x_j + q_{i}, \; \sum_{1 \le j \le n}
R_{ij} x_j) & \text{if $i \in S$,} \\
\min( \sum_{1 \le j \le n} L_{ij} x_j + q_{i}, \; \sum_{1 \le j \le n}
R_{ij} x_j) & \text{if $i \not\in S$.} 
\end{cases}
\end{equation}
\end{definition}

Observe that this definition captures systems of equations in which a min or
max operation is taken over two affine expressions in the other variables. 
Note that here we have included the additive $q$ term only in the first of the
two expressions (the second is thus linear), whereas the SDG equations have additive terms in all of the
affine expressions. This is because, for our reductions, we do not need the
second additive term.  Otherwise, this definition captures the style of Bellman
equations that are used in SDGs and other infinite games. 

We say that $x \in \reals^n$ is a solution of a $\gfixp$ system $G$ if
Equation~\eqref{eqn:maxminlin} is satisfied for all $i$. In general, however,
such an $x$ may not exist or may not be unique. To get a problem that is
equivalent to P-LCP, we need a restriction that ensures that the problem always
has a unique solution, like the P-matrix restriction on the LCP problem.

To make sure that a unique solution exists, we introduce a promise to yield the
promise problem \pfixp. As the promise we guarantee the same property that is
implied by the promise of the \plcp: For every $q'\in\reals^n$, the given
\gfixp system shall have a unique solution. We use the following restriction.

\begin{definition}\label{def:pfixp}
$\pfixp$
\begin{description}[labelindent=\deftab]
\item[Input:] A $\gfixp$ system $G = (L, R, q, S)$.
\item[Promise:] The $\gfixp$ system $(L, R, q', S)$ has a unique solution for
every $q' \in \reals^n$. 
\item[Output:] A solution $x \in \reals^n$ of $G$.
\end{description}
\end{definition}

This promise insists that the linear equations not only have a unique solution
for the given vector $q$, but that they also have a unique solution no matter
what vector $q$ is given. This is analogous to a property from SDGs:
an SDG has a unique solution no matter the rewards for the actions,
and this corresponds to the additive $q$ vector in the problem above. 

Similarly to the \lcp, where unique solutions for any right-hand side $q$ imply that the matrix $M$ is a P-matrix (and vice versa), this promise of \pfixp implies something about the involved matrices. However, for \pfixp we do not have a full characterization of the systems $(L,R,\cdot,S)$ that fulfill the promise. We only have the following rather weak necessary condition, which is however useful for several of our reductions.

\begin{lemma}\label{lem:fixpinvertible}
    In any \pfixp instance, $R-I$ is invertible.
\end{lemma}
\begin{proof}
    Given any $L,R,S$, we can pick $q$ such that $q_i=c$ iff $i\not\in S$ and $q_i=-c$ otherwise, for some very large $c>0$ dependent only on $L$. This constant is picked large enough such that any coordinate of $(L-I)x+q$ has the same sign as $q$ for all $x\in [-a,a]^d$ for some $a>0$, say $a=1$.

    For any $x\in [-a,a]^d$ we can see that if $(R-I)x=0$, then $x$ must be a solution to the \gfixp system $(L,R,q,S)$. By the promise of \pfixp, this system must have a unique solution, and thus there is at most one solution to $(R-I)x=0$ with $x\in [-a,a]^d$. We conclude that $(R-I)$ must be invertible.
\end{proof}

Since \pfixp is so similar to the \plcp, we will prove the equivalence of \plcp and \pfixp first, in the next section.

\section{Linear Complementarity Problems}\label{sec:plcp}

We begin by giving the definitions of the (binary) linear complementarity problems.

\begin{definition}
$\lcp(M, q)$
\begin{description}[labelindent=\deftab]
\item[Input:] An $n \times n$ matrix $M$ and a vector $q \in \reals^n$.

\item[Output:] Two vectors $w, z \in \reals^n$ such that
\begin{itemize}
\item $w = Mz + q$,
\item $w, z \ge 0$, and
\item $w^T \cdot z = 0$.
\end{itemize}
\end{description}
\end{definition}

\noindent We are particularly interested in the case where the input matrix $M$ is a P-matrix.

\begin{definition}
\label{def:pm}
An $n \times n$ matrix $M$ is a P-matrix if every principal minor of $M$ is
positive.
\end{definition}

One particularly interesting feature of the P-matrix Linear Complementarity
Problem is that it always has a unique solution.

\begin{theorem}[\cite{CPS09}]
\label{thm:plcp}
$M$ is a P-matrix if and only if $\lcp(M, q)$ has a unique solution for every 
$q \in \reals^n$.
\end{theorem}

There are two problems associated with P-matrix LCPs. The first
problem is a total search problem in which we must either solve the \lcp, or show that the input matrix $M$ is not a P-matrix by producing a non-positive principal minor of $M$.
The second problem (the one we consider here) is a promise problem, where the promise is that the
input matrix is a P-matrix.

\begin{definition} 
$\plcp(M, q)$
\begin{description}[labelindent=\deftab]
\item[Input:] An $n \times n$ matrix $M$ and a vector $q \in \reals^n$.
\item[Promise:] $M$ is a P-matrix.
\item[Output:] Two vectors $w, z \in \reals^n$ that are solutions of $\lcp(M, q)$.
\end{description}
\end{definition}

In the next two sections, we will show that \pfixp and \plcp are equivalent under polynomial-time many-one reductions.

\subsection{\texorpdfstring{\plcp to \pfixp}{P-LCP to P-Lin-Bellman}}
\label{sec:lcp2fixp}

Let $M \in \reals^{n \times n}$ and $q \in \reals^n$ be an instance of \plcp.
We will build a $\pfixp$ instance $G(M, q)$ in the following way. For each $i$ in the
range $1 \le i \le n$, we set
\begin{equation}
\label{eqn:plcptopfixp}
x_i = \min((Mx + x)_i + q_i, 2x_i).
\end{equation}
In other words, we set $L=M+I$, $R=2I$, and $S=\emptyset$.
We first show that every solution of $G(M, q)$ corresponds to a solution of the \lcp. We
do not need to use 
any properties of the P-matrix in this part of the proof.

\begin{lemma}
\label{lem:fp}
A vector $z \in \reals^n$ is a solution of the $\gfixp$ system $G(M, q)$ 
if and only if $z$ and $w = Mz + q$ are
a solution of $\lcp(M, q)$.
\end{lemma}
\begin{proof}
If $z$ solves \Cref{eqn:plcptopfixp} then we have
$\min(Mz + z + q, 2z) - z = 0$, and hence
that $\min(Mz + q, z) = 0$. This implies that both $z$ and $w$ are non-negative:
\begin{align*}
z &\ge \min(Mz + q, z) = 0, \\
w = Mz + q &\ge \min(Mz + q, z) = 0.
\end{align*}
Moreover, since $\min(Mz + q, z) = 0$, the complementarity condition also holds.
Hence $w$ and $z$ are a solution of the \lcp.

In the other direction, if $w$ and $z$ are solutions of the \lcp, then we argue
that $z$ must be a solution of $G(M, q)$. This is because any solution of the \lcp
satisfies $\min(Mz + q, z) = 0$ due to the complementarity condition, and so we
must have that $z$ solves \Cref{eqn:plcptopfixp}.
\end{proof}

Next we show that the promise of \pfixp is also satisfied. 

\begin{lemma}
If $M$ is a P-matrix, then $G(M, q)$ satisfies the promise of $\pfixp$. 
\end{lemma}
\begin{proof}
We must show that the $\gfixp$ system $G(M, q')$ has a unique solution
for every $q'$. This follows from \Cref{lem:fp}, which shows that $G(M, q')$ has a unique solution 
if and only if the \lcp defined by $M$ and $q'$ has a unique solution, which
follows from the fact that $M$ is a P-matrix.
\end{proof}

\subsection{\texorpdfstring{\pfixp to \plcp}{P-Lin-Bellman to P-LCP}}
\label{sec:fixp2lcp}

For this direction, we follow and generalize the approach used by Jurdzi\'nski
and Savani~\cite{JS08}. Suppose that we have a $\pfixp$ instance defined by $G
= (L, R, q, S)$. 

The reduction is based on the idea of simulating the operation $x = \max(a, b)$
by introducing two non-negative slack variables $w$ and $z$:
\begin{align*}
x - w &= a, & x - z &= b,  & w, z &\ge 0, & w \cdot z &= 0.
\end{align*}
Since $w$ and $z$ are both non-negative, any solution to the system of
equations above must satisfy $x \ge \max(a, b)$. Moreover, since the
complementarity condition insists that either $w = 0$ or $z = 0$, we have that
$x = \max(a, b)$.

An operation $x = \min(a, b)$ can likewise be simulated by
\begin{align*}
x + w &= a, & x  + z &= b,  &w, z &\ge 0, & w \cdot z &= 0.
\end{align*}

We use this technique to rewrite the system of equations given in
\eqref{eqn:maxminlin} in the following way. For each $i$ in $1 \le i \le n$ we
have the following.
\begin{align}
\label{eqn:start}
x_i + w_i &= \sum_{1 \le j \le n} L_{ij} \cdot x_j + q_{i}, &\text{if $i
\not\in S$} \\
\label{eqn:diff}
x_i - w_i &= \sum_{1 \le j \le n} L_{ij} \cdot x_j + q_{i}, &\text{if $i \in
 S$} \\
 \label{eqn:wzsum}
x_i + z_i &= \sum_{1 \le j \le n} R_{ij} \cdot x_j, &\text{if $i \not\in S$} \\
\label{eqn:wzend}
x_i - z_i &= \sum_{1 \le j \le n} R_{ij} \cdot x_j, &\text{if $i \in S$} \\
\label{eqn:pos}
w_i, z_i &\ge 0 \\
\label{eqn:cc}
w_i \cdot z_i &= 0.
\end{align}

We have already argued that this will correctly simulate the $\max$ and $\min$
operations in \Cref{eqn:maxminlin}, and so we have the following lemma. 

\begin{lemma}
\label{lem:slacksim}
$x \in \reals^n$ is a solution of $G$ if and only if $x$ is a solution of
the system defined by \Cref{eqn:start,eqn:diff,eqn:wzsum,eqn:wzend,eqn:pos,eqn:cc}.
\end{lemma}

We shall now reformulate \Cref{eqn:start,eqn:diff,eqn:wzsum,eqn:wzend,eqn:pos,eqn:cc} as
an \lcp. To achieve this we first introduce some helpful notation:
For every $n \times n$ matrix $A$, we define $\widehat{A}$ in the
following way. For each $i, j \in \{1, 2, \dots, n\}$ we let
\begin{equation*}
\widehat{A}_{ij} := \begin{cases}
\hphantom{-}A_{ij} & \text{if $i \in S$,} \\
-A_{ij} & \text{if $i \not\in S$.} 
\end{cases}
\end{equation*}
\Cref{eqn:start,eqn:diff,eqn:wzsum,eqn:wzend} can be
rewritten as the following matrix equations:
\begin{align*}
\widehat{I}x &= w + \widehat{L}x + \widehat{I}q, \\
\widehat{I}x &= z + \widehat{R}x, 
\end{align*}

Eliminating $x$ yields $(\widehat{I} - \widehat{L})(\widehat{I} - \widehat{R})^{-1}z = w +
\widehat{I}q$, which is equivalent to $w = Mz + q'$ for
\begin{align*}
M = (\widehat{I} - \widehat{L})(\widehat{I} - \widehat{R})^{-1},\;
q' = - \widehat{I}q.
\end{align*}
We rely here on the fact that by \Cref{lem:fixpinvertible}, $R-I$ and thus also 
$(\widehat{I} - \widehat{R})$ is invertible.
We have so far shown the following lemma.

\begin{lemma}
\label{lem:ftop}
A vector $x \in \reals^n$ is a solution of $G$ if and only if there are vectors $w,
z$ such that $w,z$ are a solution of $\lcp(M, q)$ and $x$, $w$, and $z$ are a solution of 
\Cref{eqn:start,eqn:diff,eqn:wzsum,eqn:wzend,eqn:pos,eqn:cc}.
\end{lemma}

To show that $M$ is a $P$ matrix, we must show that $\lcp(M, q'')$ has a unique
solution for every $q'' \in \reals^n$. By Lemma~\ref{lem:ftop}, this holds if
and only if $G(L, R, q'', S)$ has a unique solution for every $q''$, which
holds due to the $\pfixp$ promise. Our reduction is thus correct, and we have established our first main result:

\begin{theorem}
\label{thm:main}\label{thm:PFixP=PLCP}
\pfixp and \plcp are equivalent under polynomial-time many-one reductions.
\end{theorem}

\subsection{P-matrix Generalized \texorpdfstring{\lcp to \plcp}{LCP to P-LCP}}

Generalized \lcp{}s were originally introduced by Cottle and Dantzig~\cite{CD70}. A generalized \lcp instance is defined by a vertical block matrix $M$, and a vertical block
vector $q$ where
\begin{align*}
M &= \begin{bmatrix}
M_1 \\
M_2 \\
\vdots \\
M_n 
\end{bmatrix}
&
q &= \begin{pmatrix}
q_1 \\
q_2 \\
\vdots \\
q_n
\end{pmatrix}
\end{align*}
Each matrix $M_i$ has dimensionality $b_i \times n$, while each vector $q_i$
has $b_i$ dimensions. Thus, if $N = \sum_i b_i$, we have that $M \in \reals^{N
\times n}$ and $q \in \reals^N$. Given a vertical block vector $x$ with $n$
blocks, we use $x_i^j$ to refer to the $j$th element of the $i$th block of $x$. 

Given such an $M$ and $q$, the \emph{generalized linear complementarity
problem} (\glcp) is to find vectors $w \in \reals^N$ and $z \in \reals^n$ such that
\begin{itemize}
\item $w = M z + q$,
\item $w \ge 0$,
\item $z \ge 0$, and
\item $z_i \cdot \prod_{j = 1}^{b_i} w_i^j = 0 $ for all $i$.
\end{itemize}

Given a tuple $p = (p_1, p_2, \dots, p_n)$ such that each $p_i$ lies in the range
$1 \le i \le b_i$, the representative submatrix of $M$ defined by $p$ is
given by $M(p) \in \reals^{n \times n}$ and is constructed by selecting row
$p_i$ from each block-matrix $M_i$. 
We then say that $M$ is a P-matrix if all of its representative submatrices are P-matrices.

We define the P-matrix \glcp (\pglcp) as a promise problem, analogously to the \plcp: The input is a \glcp instance $(M,q)$ with the promise that $M$ is a P-matrix.  This promise again guarantees unique solutions:

\begin{theorem}[Habetler, Szanc \cite{PGLCPUniqueness}]
A vertical block matrix $M$ is a P-matrix if and only if the \glcp instance $(M,q')$ has a unique solution for every $q'\in \R^N$. 
\end{theorem}

We are now ready to show our reduction of \pglcp to \pfixp.

\begin{lemma}
\label{lem:PGLCP2PFixP}
There is a poly-time many-one reduction from \pglcp to $\pfixp$.
\end{lemma}
\begin{proof}
We can turn a P-matrix \glcp instance into a $\pfixp$ instance in much the same way as we
did for \plcp{}s in \Cref{sec:lcp2fixp}. For each $i$ in the range
$1 \le i \le n$ we set
\begin{equation}
\label{eqn:glcp}
z_i = \min\Bigl(\min \bigl\{ (M z + q)_i^j + z_i \; : 1 \le j \le b_i \bigr\},
2z_i\Bigr).
\end{equation}

We claim that any solution of the system of equations defined above yields a
solution of the \glcp. In any solution to the system we have 
\begin{equation*}
\min\Bigl(\min \bigl\{ (M z + q)_i^j \; : 1 \le j \le b_i \bigr\},
z_i\Bigr) = 0.
\end{equation*}
So if we set $w_i^j = (M z + q)_i^j$ then we have the following
non-negativities
\begin{align*}
w_i^j = (M z + q)_i^j &\ge 
\min\Bigl(\min \bigl\{ (M z + q)_i^j \; : 1 \le j \le b_i \bigr\},
z_i\Bigr) = 0. \\
z_i &\ge 
\min\Bigl(\min \bigl\{ (M z + q)_i^j \; : 1 \le j \le b_i \bigr\},
z_i\Bigr) = 0. 
\end{align*}
We can also see that the complementarity condition is satisfied, since either
$z_i = 0$ or $w_i^j = 0$ for some $j$. 

To write this as a $\gfixp$ system, we must carefully write 
\Cref{eqn:glcp} using two-input min operations. 
We will introduce a variable $x^j_i$ for each $i$ in the range $1 \le i \le n$
and each $j$ in the range $1 \le j \le b_i$. We use $\mathbf{x}^1$ to denote
the vector $(x^1_1, x^1_2, \dots, x^1_n)$.
For each
$i$, and each $j$ in the range $1 \le j < b_i$ we use the following.
\begin{equation}
\label{gp1}
x_i^j = \min \Bigl( (M \mathbf{x}^1 + q)_i^j + x^1_i, \; x_i^{j+1} \Bigr)
\end{equation}
We also use the following when $j = b_i$. 
\begin{equation}
\label{gp2}
x_i^{b_i} = \min \Bigl( (M \mathbf{x}^1 + q)_i^{b_i} + x^1_i, \; 2 x_i^{1} \Bigr)
\end{equation}

Let $G = (L, R, q, S)$ denote the resulting $\gfixp$ system.
It is now easy to see that on the one hand, for every solution $\mathbf{x}$ of $G$ the vectors $z:=\mathbf{x}^1$ and $w$ with $w_i^j=(Mz+q)_i^j$ form a solution of the input P-matrix \glcp. On the other hand, from any solution to the \glcp we can extract the vector $\mathbf{x}^1:=z$. Given this vector $\mathbf{x}^1$ there is only one way to extend this to an assignment for all $x_i^j$ that is a solution to $G$. We thus get a one-to-one correspondence between solutions to $G$ and solutions to the \glcp. Since the \glcp defined by $M$ is promised to have a unique solution for every $q'$, this thus implies that the \gfixp system $G'=(L,R,q',S)$ also has a unique solution for every $q'$. Thus, this is indeed a \pfixp instance.
\end{proof}

Combining this with the previously proven \Cref{thm:PFixP=PLCP}, we get the following corollary:

\begin{corollary}
\pglcp and \plcp are polynomial-time equivalent.
\end{corollary}

\section{Colorful Tangents and \texorpdfstring{\pfixp}{P-Lin-Bellman}}\label{sec:alphaHam}

Let us start by introducing the definitions and prior results on the $\alpha$-Ham Sandwich theorem.

\begin{definition}
    A family of point sets $P_1=\{p_{1,1},\ldots,p_{1,s_1}\},P_2,\ldots,P_d\subset\Reals^d$ is said to be \emph{well-separated} if for any non-empty index set $I\subset [d]$, there exists a hyperplane $h$ such that $h$ separates the points in $\bigcup_{i\in I}P_i$ from those in $\bigcup_{i\in [d]\setminus I}P_i$.
\end{definition}

In the setting of well-separated point sets, the classical Ham Sandwich theorem can be strengthened to the more modern $\alpha$-Ham Sandwich theorem due to Steiger and Zhao~\cite{alphaHamSandwich}, for which we need the definition of an $(\alpha_1,\ldots,\alpha_d)$-cut:

\begin{definition}\label{def:alphaCutOld}
    Given positive integers $\alpha_i\leq |P_i|$ for all $i\in [d]$, an $(\alpha_1,\ldots,\alpha_d)$-cut is a hyperplane $h$ such that $h\cap P_i\neq \emptyset$ and $|h^+\cap P_i|=\alpha_i$ for all $i\in [d]$.
\end{definition}

In this definition and in the rest of this section, $h^+$ and $h^-$ denote the two closed halfspaces bounded by a hyperplane $h$. When $h$ is a colorful hyperplane (i.e., a hyperplane containing a colorful point set, which is a set $\{p_1,\ldots,p_d\}$ with $p_i\in P_i$), its positive side is determined by the orientation of the points $p_i\in P_i$ spanning it.

\begin{theorem}[$\alpha$-Ham Sandwich theorem \cite{alphaHamSandwich}]\label{thm:alphaHam}
    Given $d$ point sets $P_1,\ldots,P_d \subset \Reals^d$ that are well-separated and in weak general position\footnote{Weak general position is a weaker version of general position. We will not go further into this since we will always ensure classical general position when invoking this theorem.}, for every $(\alpha_1,\ldots,\alpha_d)$ where $1\leq \alpha_i\leq |P_i|$, there exists a unique $(\alpha_1,\ldots,\alpha_d)$-cut.
\end{theorem}

The $\alpha$-Ham Sandwich theorem guarantees us existence and uniqueness of an $(\alpha_1,\ldots,\alpha_d)$-cut, and it is thus not too surprising that the problem of finding such a cut is contained in \Comp{UEOPL}~\cite{chiu_et_al-ComplexityAlphaHamSandwich}. 

In this section, we wish to show equivalence of this problem to \pfixp, however, we only manage this with some slight changes. Firstly, we wish to drop the assumption of weak general position, since it is not easily guaranteed when reducing from \pfixp. Because of this we need to weaken the definition of $(\alpha_1,\ldots,\alpha_d)$-cuts accordingly:

\begin{definition}\label{def:alphaCutNew}
    Given a well-separated family of point sets $P_1,\ldots,P_d\subset \Reals^d$ and $(\alpha_1,\ldots,\alpha_d)$ such that $1\leq \alpha_i\leq |P_i|$, a hyperplane $h$ is an $(\alpha_1,\ldots,\alpha_d)$-cut if for all $i\in [d]$, $h\cap P_i\neq \emptyset$ and 
    $|(h^+\setminus h)\cap P_i|+1\leq \alpha_i\leq |h^+\cap P_i|$.
\end{definition}

In other words, we allow a hyperplane to contain more than one point per color, and all the additional points may be ``counted'' for either side of the hyperplane. The following characterization of well-separation shows that even without any general position assumption,  the affine hull of every colorful set of points must have dimension at least $d-1$, and thus span a unique hyperplane. 

\begin{lemma}[Bergold et al. \cite{wellseparationTransversals}]\label{lem:stabbingFlats}
    Let $P_1,\ldots,P_k\subset \Reals^d$ be point sets. Then they are well-separated if and only if there exists no $(k-2)$-flat $h\subset\Reals^d$ (an affine subspace of dimension $k-2$) that has a non-empty intersection with each convex hull $conv(P_i)$.
\end{lemma}

Thanks to this characterization we can also show that if a hyperplane $h$ contains multiple colorful subset of points, all of them must orient $h$ the same way, and thus $h^+$ is uniquely defined even without the weak general position assumption. 

\begin{lemma}\label{lem:noOppositeOrientationsOnSameHyperplane}
    In a well-separated family of point sets, for every colorful hyperplane $h$, every colorful set of points in $h$ is oriented the same way.
\end{lemma}
\begin{proof}
    Towards a contradiction, consider two colorful sets of points in $h$ that are oriented differently. Since we can go from one to the other by replacing points of each color one by one, we know that the orientation must change when replacing some point. Thus, there must also be two colorful point sets that are oriented differently and that only differ within one point, i.e., one is spanned by $Q\cup\{p\}$ and the other by $Q\cup\{p'\}$. For them to be oriented differently, $p$ and $p'$ need to lie on opposite sides of the $(d-2)$-flat spanned by $Q$. But then, this flat stabs the convex hull of all point sets, and the points cannot be well-separated.
\end{proof}

Note that the $\alpha$-Ham Sandwich theorem (\Cref{thm:alphaHam}) generalizes directly to the setting where the assumption of weak general position is dropped and \Cref{def:alphaCutOld} is replaced by \Cref{def:alphaCutNew}.

\begin{theorem}\label{thm:aHSwithTieBreaking}
    Given a well-separated family of point sets $P_1, P_2,\ldots,P_d\subset\Reals^d$ and $(\alpha_1,\ldots,\alpha_d)$ for $1\leq \alpha_i\leq |P_i|$, there exists a unique hyperplane $h$ that is an $(\alpha_1,\ldots,\alpha_d)$-cut according to \Cref{def:alphaCutNew}.
\end{theorem}
\begin{proof}
    This simply follows from perturbing the points and applying \Cref{thm:alphaHam}.
\end{proof}

As our second change, we have to strengthen well-separation to \emph{strong well-separation}. We illustrate well-separation and strong well-separation in \Cref{fig:SWSvsWS}.

\begin{definition}
    A family of point sets $P_1=\{p_{1,1},\ldots,p_{1,s_1}\},P_2,\ldots,P_d\subset\Reals^d$ is said to be \emph{strongly well-separated} if the point sets $P'_1,\ldots,P'_d$ obtained by projecting $P_1,\ldots,P_d$ to the hyperplane spanned by $p_{1,1},\ldots,p_{d,1}$ are well-separated.
\end{definition}

\begin{figure}[hbt]
    \centering
    \begin{subfigure}{0.4\textwidth}
    \includegraphics{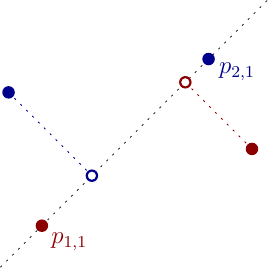}
    \end{subfigure}
    \hfill
    \begin{subfigure}{0.4\textwidth}
    \includegraphics{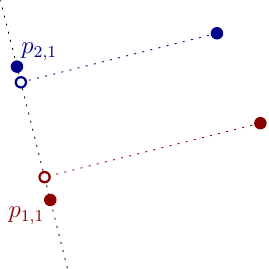}
    \end{subfigure}
    \caption{The point sets (dots) on the left are well-separated but not strongly well-separated, since their projections (circles) are not well-separated. The point sets on the right are strongly well-separated.}
    \label{fig:SWSvsWS}
\end{figure}

Under the assumption of strong well-separation, we have an even stronger version of \Cref{lem:noOppositeOrientationsOnSameHyperplane}:

\begin{lemma}\label{lem:noOppositeOrientations}
    Let $P_1,\ldots,P_d\subset\R^d$ be a strongly well-separated family of point sets. Let $h$ be a colorful hyperplane and $v$ its positive-side normal vector. Let $w$ be the positive-side normal vector of the colorful hyperplane spanned by the points $p_{1,1},\ldots,p_{d,1}$. Then $v^Tw>0$.
\end{lemma}
\begin{proof}
    Towards a contradiction, consider two colorful hyperplanes $h$ and $h'$ with $v^Tw>0$ and $v'^Tw<0$ (note that equality is impossible due to the strong well-separation assumption). Since we can go from one to the other by replacing points of each color one by one, we know that the sign of the dot-product must change when replacing some point. Thus, there must also be two colorful hyperplanes with different sign that only differ within one point, i.e., one is spanned by $Q\cup\{p\}$ and the other by $Q\cup\{p'\}$. For them to be oriented differently, in the projection of $Q\cup\{p,p'\}$ onto the hyperplane spanned by $p_{1,1},\ldots,p_{d,1}$, $p$ and $p'$ need to lie on opposite sides of the $(d-2)$-flat spanned by $Q$. But then, this flat stabs the convex hull of all projected point sets, and thus these projected point sets cannot be well-separated, which by definition contradicts strong well-separation of $P_1,\ldots,P_d$.
\end{proof}

We are now ready to introduce the problem we wish to study.
\begin{definition} 
$\swsHS(P_1,\ldots,P_d,\alpha_1,\ldots,\alpha_d)$
\begin{description}[labelindent=\deftab]
\item[Input:] $d$ point sets $P_1,\ldots,P_d\subset\Reals^d$ and a vector $(\alpha_1,\ldots,\alpha_d)$, where $\alpha_i\in \{1,|P_i|\}$.
\item[Promise:] The point sets $P_1,\ldots,P_d$ are strongly well-separated.
\item[Output:] An $(\alpha_1,\ldots,\alpha_d)$-cut.
\end{description}
\end{definition}
We depart from the \emph{Ham Sandwich} name for this problem since for $\alpha_i\in\{1,|P_i|\}$, the $\alpha$-``cuts'' are really just colorful tangent hyperplanes, as can be seen in \Cref{fig:ColTangentsExample}.

\begin{figure}[hbt]
    \centering
    \includegraphics{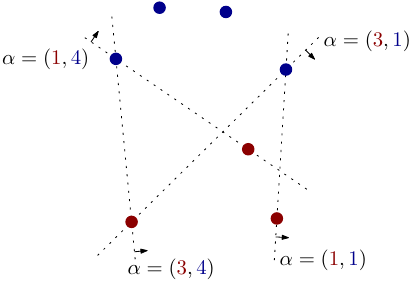}
    \caption{Two well-separated point sets and the solution $\alpha$-cuts to all four possible $\alpha$-vectors. Note that the positive side of a colorful line is considered to be the right side when orienting the line from the red point to the blue point.}
    \label{fig:ColTangentsExample}
\end{figure}

Similarly to \pglcp and its binary variant \plcp, we also define \swsTwoHS as the restriction of \swsHS to inputs where $|P_i|=2$ for all $i\in [d]$. We now wish to prove these two problems polynomial-time equivalent to \pfixp. To achieve this, we reduce \swsHS to \pfixp, and \pfixp to \swsTwoHS, with the reduction by inclusion from \swsTwoHS to \swsHS closing the cycle.

For our reductions, we will use the concept of \emph{point-hyperplane duality}, as described by Edelsbrunner in \cite[p.13]{edelsbrunner1987geometry}. Point-hyperplane duality is a bijective mapping from all points in $\R^d$ to all non-vertical hyperplanes in $\R^d$. A point $p=(p_1,\ldots,p_d)$ is mapped to the hyperplane $p^*=\{x\in\R^d\;\vert\; 2p_1x_1 + \ldots 2p_{d-1}x_{d-1}-x_d = p_d\}$. The hyperplane $p^*$ is called \emph{the dual of $p$}, and for a non-vertical hyperplane $h$, its dual $h^*$ is the unique point $p$ such that $p^*=h$.

Point-hyperplane duality has the nice properties of \emph{incidence preservation} and \emph{order preservation}: for any point $p$ and non-vertical hyperplane $h$, we have that $p\in h$ if and only if $h^*\in p^*$, and furthermore we have that $p$ lies above $h$ if and only if $h^*$ lies above $p^*$.

We are now ready to begin presenting our reductions. For all of these reductions, we use the following crucial yet simple observation on strongly well-separated point sets:

\begin{observation}\label{obs:stronglyWSimpliesWSforallshifts}
    If a set of point sets $P_1,\ldots,P_d$ is strongly well-separated, then every set of point sets $P'_1,\ldots,P'_d$ obtained by moving points $p_{i,j}$ ($j\neq 1$) orthogonally to the hyperplane spanned by $p_{1,1},\ldots,p_{d,1}$ is also strongly well-separated.
\end{observation}

The general idea of the reductions is to represent the linear inputs to a $\min$ or $\max$ operation in a \gfixp system by a point $p_{i,1}$ and the affine inputs by a point $p_{i,j}$ ($j\neq 1$).

\subsection{\swsHS to \pfixp}
As a warm-up, we first only reduce from the two-point version.

\begin{lemma}\label{lem:SWS2PaHS-to-PFixP}
    There is a poly-time many-one reduction from \swsTwoHS to \pfixp.
\end{lemma}
\begin{proof}
    We first linearly transform our point sets such that the plane through $p_{1,1},\ldots,p_{d,1}$ is mapped to the horizontal plane $\{(x_1,\ldots,x_d)\in \R^d\;\vert\; x_d=0\}$. Without loss of generality, we assume that after this transformation the colorful hyperplane $h$ spanned by $p_{1,1},\ldots,p_{d,1}$ is oriented upwards, i.e., its positive halfspace contains $(0,\ldots,0,+\infty)$.
    
    Then, we apply point-hyperplane duality. Each point $p_{i,j}$ becomes a hyperplane $h_{i,j}=p_{i,j}^*$, which can be described by some function $h_{i,j}:\R^d\rightarrow \R$ such that a point $x\in \R^d$ lies strictly above $h_{i,j}$ if $h_{i,j}(x)>0$, and on $h_{i,j}$ if $h_{i,j}(x)=0$. For the hyperplanes $h_{i,1}$ dual to our points $p_{i,1}$, we furthermore have that they go through the origin, i.e., $h_{i,j}$ is a linear function (and not only an affine function; it has no additive constant term).

    Before applying duality, the desired $(\alpha_1,\ldots,\alpha_d)$-cut hyperplane is a hyperplane $h$ containing at least one point $p_i'$ per set $P_i$, such that the other point of $P_i$ lies on or above $h$ if and only if $\alpha_i=|P_i|$, and on or below $h$ if and only if $\alpha_i=1$. In the dual, this now means that we need to find a point $p$ which lies on at least one of the hyperplanes $h_{i,1},h_{i,2}$ for each $i$, and such that it lies above (below) or on both of these hyperplanes if $\alpha_i=|P_i|$ ($\alpha_i=1$).
    This can be described by the equations $\min\{h_{i,1}(x),h_{i,2}(x)\}=0$ if $\alpha_i=|P_i|$ (and $\max$ instead of $\min$, otherwise). This can of course be rewritten as $x_i=\min\{h_{i,2}(x)+x_i, h_{i,1}(x)+x_i\}$. Note again that the second input to this minimum is \emph{linear} in $x$ (there is no additive term). With one such constraint each per point set $P_i$, we thus get a system of $d$ equations over $d$ variables, which together form a \gfixp system.

    It remains to check whether this system is also a \pfixp instance. We first see that changing $q$ simply corresponds to moving the hyperplanes $h_{i,2}$ in a parallel fashion (i.e., without changing their normal vectors), which in the primal corresponds to moving them in direction $x_d$. By \Cref{obs:stronglyWSimpliesWSforallshifts} and \Cref{thm:aHSwithTieBreaking}, we always have unique $(\alpha_1,\ldots,\alpha_d)$-cuts in this modified family of point sets, and thus the \gfixp instance has a unique solution for all $q'$.
\end{proof}

We now generalize this reduction in a very similar fashion as we generalized the reduction from \plcp to \pfixp to work for \pglcp instead.

\begin{theorem}\label{thm:SWSaHS2PFixP}
    There is a poly-time many-one reduction from \swsHS to \pfixp.
\end{theorem}
\begin{proof}
    The proof works exactly the same way as the previous proof of \Cref{lem:SWS2PaHS-to-PFixP}. The only difference is that the equations that need to be encoded are of the form 
    \begin{equation}\label{eqn:minimumManySWSHS}
        x_i=\min\{h_{i,1}(x)+x_i,\ldots,h_{i,|P_i|}(x)+x_i\},
    \end{equation}
    where still only $h_{i,1}(x)$ is guaranteed to be a linear function. To encode this as a \gfixp system, we apply the same trick as in the proof of \Cref{lem:PGLCP2PFixP} to split the multi-input minimum into multiple two-input minima.

    For each variable $x_i$, we introduce $|P_i|-1$ helper variables $z_i^1,\ldots,z_i^{|P_i|-1}$. The first helper variable $z_i^1$ will replace $x_i$, and we thus write $\mathbf{z^1}$ to mean the vector $(z_1^1,\ldots,z_d^1)$ replacing $x$. Then, we express \Cref{eqn:minimumManySWSHS} using the following $|P_i|$ equations: For $1\leq j\leq |P_i|-2$, we add
    $$z_i^j = \min\{h_{i,|P_i|-j+1}(\mathbf{z^1}) + z_i^1, z_i^{j+1}\} $$
    and finally we also add
    $$z_i^{|P_i|-1} = \min\{ h_{i,2}(\mathbf{z^1}) + z_i^1, h_{i,1}(\mathbf{z^1}) + z_i^1\}.$$
    It is easy to see that in any solution, the vector $\mathbf{z^1}$ must be a valid solution $x$ for the \Cref{eqn:minimumManySWSHS}. Since by \Cref{obs:stronglyWSimpliesWSforallshifts,thm:aHSwithTieBreaking} there is a unique such $x$ for all $q'$, and since given $\mathbf{z^1}$ all the other helper variables $z_i^j$ are determined uniquely, this \gfixp-system has a unique solution for all $q'$. 
\end{proof}

\subsection{\pfixp to \swsTwoHS}
To complete our cycle of reductions, we need to reduce from \pfixp to \swsTwoHS. For this reduction we basically perform the process from the proof of \Cref{lem:SWS2PaHS-to-PFixP} in reverse.

\begin{theorem}\label{thm:PFixP-to-SWS2PaHS}
    There is a poly-time many-one reduction from \pfixp to \swsTwoHS.
\end{theorem}
\begin{proof}
    Recall the reductions between \pfixp and P-LCP from \Cref{sec:fixp2lcp,sec:lcp2fixp}. If a \pfixp instance is first reduced to P-LCP and then back to \pfixp, we end up with an instance with $S=\emptyset$. To prove the desired theorem, we thus only need to reduce such instances $(L,R,q,S=\emptyset)$ to \swsTwoHS.

    We consider the process from the proof of \Cref{lem:SWS2PaHS-to-PFixP} above in reverse. For each equation 
    \[x_i=\min\{l_{i,1}x_1+\ldots+l_{i,d}x_d+q_i, r_{i,1}x_1+\ldots+r_{i,d}x_d\}\]
    in the given \gfixp system we interpret both inputs to the minimum as a hyperplane equation. Thus we get hyperplanes $h_{i,2}(x)=l_{i,1}x_1+\ldots+(l_{i,i}-1)x_i+\ldots+l_{i,d}x_d + q_i$ and $h_{i,1}(x)=r_{i,1}x_1+\ldots+(r_{i,i}-1)x_i+\ldots+r_{i,d}x_d$. Clearly all $h_{i,1}$ go through the origin, since there is no additive constant. We now apply point-hyperplane duality to arrive back in the primal view, where we get our point sets $P_1,\ldots,P_d$. Without loss of generality we can again assume that in the resulting point set, the colorful hyperplane spanned by $p_{i,1},\ldots,p_{d,1}$ is horizontal and oriented upwards. If the point sets $P_1,\ldots,P_d$ are strongly well-separated, we can see by the same arguments as in the proof of \Cref{lem:SWS2PaHS-to-PFixP} and by \Cref{lem:noOppositeOrientations} that they form an instance of \swsTwoHS and every
    $(2,\ldots,2)$-cut
    corresponds to a solution of the given \pfixp instance and vice versa.

    It thus only remains to prove that these point sets are indeed strongly well-separated. To do this, we use the promise that the \pfixp instance has a unique solution for all $q'$. Recall from the proof of \Cref{lem:SWS2PaHS-to-PFixP} that changing the $i$-th entry of $q$ corresponds to moving the point $p_{i,2}$ orthogonally to the hyperplane $h$ spanned by $p_{1,1},\ldots,p_{d,1}$ (note that by \Cref{lem:fixpinvertible}, these points actually must span a hyperplane). Since we can thus move these points (individually) by changing $q$, we call the points $p_{i,2}$ \emph{movable}, while the points $p_{i,1}$ are \emph{non-movable}. We aim to show that the point set family obtained from $q'=0$ is well-separated. This then implies that the point set family obtained from any $q'$ (including our particular input $q$) is strongly well-separated, since $q'=0$ corresponds to all the movable points being projected onto the hyperplane $h$.

    We assume towards a contradiction that the point sets for $q'=0$ are \emph{not} well-separated. Then, by \Cref{lem:stabbingFlats}, there must be a $(d-2)$-flat $f$ contained in the hyperplane $h$ which stabs the convex hulls of all point sets. Note that since we only have two points per color, these convex hulls are just segments.
    
    To find a contradiction, we want to find some alternative lifting vector $q^*$ for which there exist two distinct solutions to the \gfixp system $(L,R,q^*,\emptyset)$. This would then of course be a contradiction to the fact that the given \gfixp system is a \pfixp instance. Note that a solution to the \gfixp system is a \emph{colorful lower tangent}, i.e., a  hyperplane $h'$ such that for each color, at least one point lies on the hyperplane and the other lies above.\footnote{In this setting of \gfixp systems, ``above'' means that the point lies in the halfspace bounded by $h'$ that contains $(0,\ldots,0,+\infty)$. We are not using the notion of orientation obtained from the orientation of some colorful set of points on the hyperplane. In fact, since we are not assuming well-separation, that notion could be inconsistent depending on the set of points we pick.}

    To see that we can find such a $q^*$ we distinguish two cases based on which non-movable points lie on which side of $f$, i.e., in which of the two components of  $h\setminus f$. We call these two components (open halfspaces within $h$) $f^-$ and $f^+$. Note that points may also lie on $f$ itself, however since the non-movable points span $h$, at least one non-movable point does not lie on~$f$.
    
    \begin{enumerate}
        \item Some side, w.l.o.g. $f^+\cup f$ contains all non-movable points: Then we can consider two hyperplanes, $h$ spanned by all the non-movable points, and $h'\supset f$ which lies strictly below $h$ on the side $f^+$. Then we can move the movable points in $f^-$ such that they all lie on $h'$ and above or on $h$. Since $f$ is stabbing, both of these hyperplanes are colorful, and they are clearly lower tangents, as illustrated in \Cref{fig:foldingcase1}. 
        \begin{figure}
            \centering
            \includegraphics{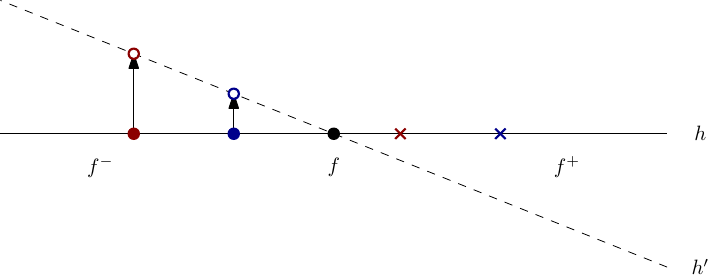}
            \caption{Case 1 of the proof of \Cref{thm:PFixP-to-SWS2PaHS}: all non-movable points (crosses) lie in $f^+\cup f$. Then the movable points (represented by solid dots) can be moved up so that both $h$ spanned by the non-movable points as well as $h'$ spanned by the lifted movable points (circles) are colorful lower tangents.}
            \label{fig:foldingcase1}
        \end{figure}
        \item Both $f^-$ and $f^+$ contain at least $1$ non-movable point. In this case let $A$ be the set of non-movable points in $f^-\cup f$, and $B$ the set of non-movable points in $f^+\cup f$. Note that $|A|,|B|\leq d-1$.
        We consider the flats $\alpha := \aff(A)$ and $\beta := \aff(B)$ spanned by these sets. 
        We now consider the intersections $\alpha_f:=\alpha\cap f$ and $\beta_f:=\beta\cap f$. Let $g\subset f$ be the affine hull of $\alpha_f\cup \beta_f$. 
                
        We now wish to move our movable points such that 
        \begin{itemize}
            \item all points previously in $f^-\cup f$ lie on a common hyperplane $h_-$,
            \item all points previously in $f^+\cup f$ lie on a common hyperplane $h_+$,
            \item $h_-\cap h_+$ is a $(d-2)$-dimensional flat which projected onto $h$ equals $f$,
            \item $h_-$ and $h_+$ are both colorful lower tangents.
        \end{itemize}
        To achieve this, we first define the two hyperplanes $h_-$ and $h_+$, and then simply adjust our movable points to lie on the correct hyperplanes. The first hyperplane $h_-$ is defined by rotating around a $(d-2)$-dimensional rotation axis containing $A$ and $g$. To see that such a rotation axis exists, we  show that $\dim(\aff(A\cup g))\leq d-2$. First, we observe that $\dim(\alpha)\leq |A|-1$ and $\dim(\beta_f)\leq |B|-2$. Furthermore, $\dim(\alpha\cap \beta_f)=\dim(\alpha_f\cap\beta_f)\geq |A\cap B|-1$. Thus, $\dim(\aff(\alpha\cup \beta_f))\leq (|A|-1)+(|B|-2) -(|A\cap B|-1)=d-2$. Since $\aff(\alpha\cup \beta_f)=\aff(A\cup g)$, our desired rotation axis exists.
        Since this rotation axis contains $A$, we also have $A\subset h_-$, and thus all non-movable points in $f^-\cup f$ lie on $h_-$, as desired. Similarly, $h_+$ is $h$ rotated around an axis containing $B$ and $g$. Note that we pick our two rotation axes such that their intersection is contained in $f$. The direction of these rotations is picked such that some fixed non-movable point in $f^+$ lies above $h_-$, and some fixed non-movable point in $f^-$ lies above $h_+$. The third constraint, $h_-\cap h_+$ projecting onto $f$, links the necessary angle of rotation between the two hyperplanes\footnote{There always exist angles to fulfill this constraint: Firstly, note that the projection of $h_-\cap h_+$ onto $h$ always contains the intersection of the two rotation axes, which is in $f$. Secondly, note that we can make the projected intersection lie in each rotation axis by making either $h_-$ or $h_+$ vertical. By continuity, any projected intersection between those two extremes can also be realized, thus in particular $f$.}. This is illustrated in \Cref{fig:foldingcase2}.

        \begin{figure}[hbt]
            \centering
            \includegraphics{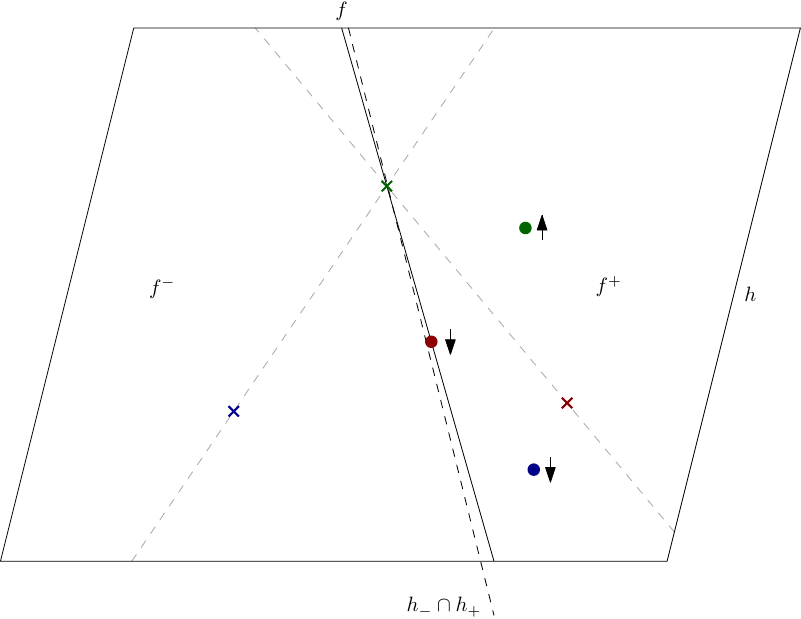}
            \caption{Case 2 of the proof of \Cref{thm:PFixP-to-SWS2PaHS}: at least one non-movable point (crosses) lies in each side $f^+,f^-$. Then we can find two hyperplanes $h_+,h_-$ by rotating $h$ along the two dashed axes. The movable points (solid dots) can be moved up or down to lie on these rotated hyperplanes.}
            \label{fig:foldingcase2}
        \end{figure}
        
        It only remains to show that both $h_-$ and $h_+$ are colorful lower tangents: Each hyperplane is clearly colorful since it contains all points in $f^-\cup f$ (or $f^+\cup f$), and $f$ is stabbing the convex hulls of all $P_i$. Furthermore, to see that the hyperplanes are lower tangents, we view the point set projected onto $h$. Since $h_-\cap h_+$ projects onto $f$, $h_-$ lies above $h_+$ within all of $f^-$, and $h_+$ lies above $h_-$ within all of $f^+$. Since all points in $f^+$ lie on $h_+$ and all points in $f^-$ lie on $h_-$, both hyperplanes must be lower tangents.
    \end{enumerate}
    Clearly these two cases cover all possibilities. We have thus shown that under the assumption that the point sets for $q'=0$ are not well-separated, the given \gfixp system cannot have been a \pfixp instance, since we have constructed a $q^*$ for which $(L,R,q^*,S)$ has multiple solutions. Due to this contradiction we conclude that the point sets obtained from $q'=0$ must be well-separated, which in turn proves that the point sets obtained from $q$ are \emph{strongly} well-separated.
\end{proof}

\subsection{Well-Separation versus Strong Well-Separation}

To end our treatise on the colorful tangent problems, we want to note that the assumption of strong well-separation is not too strong, at least combinatorially:

\begin{theorem}\label{thm:samecombinatorics}
    For every family of well-separated point sets, there exists a family of strongly well-separated point sets with the same combinatorial structure, i.e., the two families have the exact same order type.
\end{theorem}
To prove \Cref{thm:samecombinatorics}, we first need the following lemma:
\begin{lemma}\label{lem:additionalPoint}
    Let $P_1,\ldots,P_{d+1}\subset\R^d$, with $P_{d+1}=\{p\}$ be a family of well-separated point sets, and let $h$ be a hyperplane with $p\not\in h$. Let $P'_1,\ldots,P'_d$ be the point sets obtained by projecting all the points $q\in \bigcup_{i\in [d]}P_i$ onto $h$ by replacing $q$ with the point $q':=h\cap \overline{qp}$. Then $P'_1,\ldots,P'_d$ are also well-separated.
\end{lemma}
\begin{proof}
    Towards a contradiction, assume $P'_1,\ldots,P'_d$ are not well-separated. Then, by \Cref{lem:stabbingFlats}, there exists a $(d-2)$-flat $f$ which intersects with all of their convex hulls. Let $f'$ be the affine hull of $f$ and $p$. Clearly, $f'$ is a $(d-1)$-flat, and must intersect all convex hulls of $P_1,\ldots,P_d,P_{d+1}$. Thus, these sets would not have been well-separated to begin with, a contradiction.
\end{proof}

\begin{proof}[Proof of \Cref{thm:samecombinatorics}]
    We wish to prove existence of a ray $r$, such that for every $p\in r$, our point sets remain well-separated with $P_{d+1}=\{p\}$ added as one more color. To do this, we look at the set $H$ of up to $2^d$ colorful hyperplanes that are $(\alpha_1,\ldots,\alpha_d)$-cuts for $\alpha_i\in\{1,|P_i|\}$. We furthermore look at the positive halfspaces (direction induced by the orientation of the points) of these hyperplanes. If a point $p$ lies in the intersection of all of these halfspaces, for any $I\subseteq [d]$ we can separate $\{p\}\cup\bigcup_{i\in I}P_i$ from $\bigcup_{i\not\in I}P_i$ by simply taking the correct hyperplane in $H$. Thus, to prove existence of the ray $r$, we simply need to show that in the hyperplane arrangement $H$, the cell which corresponds to the positive side of all hyperplanes is non-empty and unbounded. Then we can simply embed the ray $r$ in this cell.

    We first consider a larger set of hyperplanes, namely the set $T$ of all \emph{transversals}. A transversal is a hyperplane which has a non-empty intersection with $conv(P_i)$ for all $i\in [d]$. We obtain a halfspace from a transversal $t$ by orienting it based on some arbitrary choice of one point in each intersection $t\cap conv(P_i)$. By \Cref{lem:stabbingFlats} switching each color with its convex hull cannot destroy well-separation, thus \Cref{lem:noOppositeOrientationsOnSameHyperplane} applies to the family $conv(P_1),\ldots,conv(P_d)$. Using this we show that no two of these halfspaces have exactly opposite normal vectors: By the lemma, no transversal can be oriented in both directions, thus if there were two such halfspaces, their bounding transversals would have to be distinct yet parallel. We can move from one of these transversals to the other continuously in a parallel fashion. During this process, the hyperplane remains a transversal by convexity. At some point, the orientation of this transversal must switch, which is however impossible since no transversal can have both orientations by \Cref{lem:noOppositeOrientationsOnSameHyperplane} (it is also impossible for the orientation to switch because the colorful points on one transversal do not span the hyperplane, since that would violate well-separation by \Cref{lem:stabbingFlats}).
    
    Furthermore, this set of normal vectors is closed under positive combinations\footnote{Let $t_1,t_2$ be two transversals, let $v$ be a positive combination of their normal vectors. For each $i\in[2],j\in[d]$, pick some point $q_{i,j}\in t_i\cap conv(P_j)$. Then, for some choice of $q_j'\in conv(q_{1,j},q_{2,j})$ for all $j\in [d]$, the transversal going through all the $q_j'$ has the desired normal vector $v$.}. Thus, the set of normal vectors of the positive halfspaces of $T$ forms a cone which is not equal to $\R^d$.

    We turn our attention back to $H$. Note that the positive halfspaces considered in $H$ are a subset of the positive halfspaces considered in $T$. 
    We first show that the desired cell cannot be bounded but non-empty: If it was a bounded cell, there would be some subset of halfspaces $H'$ which bound a polytopal cell, with all normal vectors pointing into the cell. The set of normal vectors of the facets of a polytope span all of $\R^d$ with a positive combination, which contradicts that the set of normal vectors in $T$ is not all of $\R^d$. Thus the cell cannot be bounded.

    Next, we show that the cell cannot be empty. To prove this, we show that any $d+1$ halfspaces of $H$ have a non-empty intersection. Then, by Helly's theorem, all of them have a non-empty intersection. The only way for $d+1$ non-parallel halfspaces to have an empty intersection is to have the hyperplanes bound some simplex, with all normal vectors pointing outwards. Again, we have that these normal vectors would then positively span all of $\R^d$, which is again a contradiction.

    We thus conclude that the ray $r$ exists. We now pick the point $p$ to be the point at infinity on $r$, and we pick the hyperplane $h$ through our points $p_{1,1},\ldots,p_{d,1}$. By \Cref{lem:additionalPoint}, we can project all our points onto $h$ in a parallel (but not yet orthogonal) way and remain well-separated. There now exists an affine transformation we can apply to our input point set which does not move the points $p_{1,1},\ldots,p_{d,1}$, but which transforms the direction of $r$ to be orthogonal to the hyperplane through these points. Then, this projection becomes orthogonal, proving strong well-separation of the transformed point sets. Note that the affine transformation does not change the order type of the point sets.
\end{proof}

Note that this is merely an existence result, finding the ray $r$ for the required affine transformation or directly finding a strongly well-separated point set family with the same order type is most likely computationally hard.

\section{\GridUSO and \CubeUSO}\label{sec:gridtocube}
We saw in the two previous sections that both \pglcp and \swsHS can be reduced to their respective ``binary'' variants, \plcp and \swsTwoHS. In this section we discuss grid USOs and cube USOs, which are a combinatorial framework that can be used to model all the problems studied in the sections above as well as many more algebraic and geometric problems. Similarly to the previous sections, cube USOs are a  restriction of grid USOs to the ``binary'' case.

\subsection{Definitions}\label{sec:USOdefs}

We define $C$ as a \dimension-dimensional hypercube. 
We say $V(C) := \{0,1\}^\dimension$ and $ \{J, K\} \in  E(C)$ iff $\abs{J \xor K} = 1$, where \xor is the bit-wise xor operation (i.e., addition in $\mathbb{Z}_2^d$) and $|\cdot|$ counts the number of $1$-entries. For notational simplicity, in this section we use the same name both for a bitvector in $\{0,1\}^\dimension$ and for the set of dimensions $i\in [d]$ in which the vector is $1$. For example for $J,K\in\{0,1\}^\dimension$ we write $J \subseteq K$ if for all $i \in [\dimension]$ with $J_i = 1$ we also have $K_i = 1$.

The orientation of the edges of the cube $C$ is given by an \emph{orientation} function  $\orientationCube: V(C) \rightarrow \{0,1 \}^\dimensionA$, where $\orientationCube(J)_i=0$ means $J$ has an incoming edge in dimension $i$ and $\orientationCube(J)_i=1$ is an outgoing edge in dimension $i$.

\begin{definition}
\label{def:USO}
An orientation is a \emph{unique sink orientation (USO)} if and only if every induced subcube has a unique sink.
\end{definition}

The common search problem version of this problem is to find the global sink of the cube, given the function \orientationCube as a boolean circuit. Note that it is \Comp{co-NP}-complete to test whether a given orientation is USO \cite{gaertner2015recognizing}.
We consider the promise version \CubeUSO, which was one of the first search problems proven to lie in the complexity class \Comp{Promise-UEOPL} \cite{fearnley2020ueopl}.

\begin{definition} \label{def:CubeUSOSearchProblem}
\CubeUSO
\begin{description}[labelindent=\deftab]
\item[Input:] A circuit computing the orientation function \orientationCube on a \dimension-dimensional cube $C$.
\item[Promise:] \orientationCube is a Unique Sink Orientation.
\item[Output:]  A vertex \f{J \in V(C)} which is a sink, i.e., \f{\forall i \in [\dimensionA]: \orientationCube(J)_i = 0}.
\end{description}
\end{definition}

While the $d$-dimensional hypercube is the product of $d$ copies of $K_2$ (the complete graph on two vertices), a \emph{grid graph} is the product of complete graphs of arbitrary size:
A \dimensionA-dimensional grid  graph \Grid is given by $\partitionLength_1, \dots, \partitionLength_\dimensionA \in \Naturals^+$:
\begin{align*}
V(\Grid) &:= \{0, \dots, \partitionLength_1\}\times \dots \times \{0, \dots, \partitionLength_\dimensionA\},\\
E(\Grid) &:= \{\{v, w\} \mid v, w \in V(\Grid), \exists! i \in [\dimensionA]: v_i \neq w_i \}.
\end{align*}
We say the grid has $\dimensionA$ \emph{dimensions} and each dimension $i$ has $n_i+1$ \emph{directions}.

The subgraph $\Grid'$ of \Grid induced by the vertices $V(\Grid') = N_1 \times \dots \times N_\dimension$ for non-empty $N_i \subseteq \{0, \dots, \partitionLength_i\}$ 
is called an \emph{induced subgrid} of $\Grid$.
Note that if for some $i$ we have $\abs{N_i}=1$, the induced subgrid loses a dimension.
If $\abs{N_i}=1$ for all $i\in [d]$ except one, we say that the induced subgrid $\Grid'$ is a \emph{simplex}. A simplex is a complete graph $K_{n_j+1}$ for some $j\in [d]$.

The orientation of a grid is given by the \emph{outmap function}, which assigns each vertex a binary vector that encodes whether its incident edges are incoming or outgoing. More formally, the outmap function is a function $\sigma:V(\Grid)\rightarrow \{0,1\}^{n_1+\ldots+n_d}$, where $\sigma(v)_{n_1+\ldots+n_i+j}=1$ denotes that the edge from $v$ to its $j$-th neighbor $w$ in dimension $i+1$ is outgoing, i.e., oriented from $v$ to $w$. Note that any circuit computing $\sigma$ has $n_1+\ldots+n_d$ outputs, and is thus of size $\Omega(n_1+\ldots+n_d)$.

For notational convenience, we denote the entry of $\sigma(v)$ relevant to the edge $\{v,w\}$ by $\sigma(v,w)$. In other words, for each pair of vertices $v,w$ such that $\{v,w\}\in E(\Grid)$, $\sigma(v,w)=1$ iff the edge between $v$ and $w$ is oriented towards $w$.

\begin{definition}{(Gärtner, Morris, Rüst \cite{gaertner2008grids})}
\label{def:GridUSO}
An orientation of a grid graph is a \emph{unique sink orientation (USO)} if and only if every induced subgrid has a unique sink.
\end{definition}

A unique sink orientation of a simplex with \dimensionB vertices is equivalent to a permutation of \dimensionB elements, i.e., every unique sink orientation of one simplex can be given by a total order of its vertices.
The minimum element of the order is the source, the maximum element is the sink.

Since every cube is also a grid, the problem of checking whether \orientationGrid is USO is once again \Comp{co-NP}-hard.
So again, we consider the promise search version of this problem:

\begin{definition} \label{def:GridUSOSearchProblem}
\GridUSO
\begin{description}[labelindent=\deftab]
\item[Input:]  A \dimensionA-dimensional grid \f{\Grid = (\partitionLength_1, \dots, \partitionLength_\dimensionA)} and a circuit computing its outmap \f{\orientationGrid}.
\item[Promise:] \orientationGrid is a Unique Sink Orientation.
\item[Output:]  A sink, i.e., a vertex \f{v \in  V(\Grid)} s.t. \f{\forall w \in V(\Grid)} with $\{v, w\} \in E(\Grid) : \orientationGrid(v, w) = 0$.
\end{description}
\end{definition}

Just like \CubeUSO, \GridUSO also lies in the search problem complexity class \Comp{Promise-UEOPL}~\cite{borzechowski2021ueopl}.

\subsection{Colorful Tangents, \plcp, and USO}
It is well known that \plcp reduces to \CubeUSO~\cite{stickney1978digraph} and \pglcp reduces to \GridUSO~\cite{gaertner2008grids}.
Since \swsTwoHS and \swsHS can both be reduced to \plcp and \pglcp respectively, they also reduce to \CubeUSO and \GridUSO respectively. 
However, we show a direct and straightforward reduction from \swsTwoHS to \CubeUSO, and from \swsHS to \GridUSO. These direct reductions do not require strong well-separation, only classical well-separation.

\begin{lemma}\label{lem:swsTwoHs2cubeUSO}
    Assuming general position of the input points, finding an $\alpha$-cut in a well-separated point set family $P_1,\ldots,P_d\subset \R^d$ for $\alpha_i\in\{1,|P_i|\}$ can be reduced to \GridUSO in polynomial time. The reduction goes to \CubeUSO if for all $i$, $|P_i|=2$.
\end{lemma}

We first prove the simpler case of \Cref{lem:swsTwoHs2cubeUSO}, when $|P_i|=2$ for all $i$.

\begin{proof}[Proof of \Cref{{lem:swsTwoHs2cubeUSO}} for $|P_i|=2$]
    We want to construct an orientation function $\orientationCube$ on a $d$-dimensional hypercube that is USO such that we can derive the $\alpha$-cut of the point sets from its sink.

    Each vertex $v\in V(C)$ corresponds to a colorful choice of points. If $v_i=0$, then of color $i$ we choose $p_{i, 1}$, if $v_i=1$, we choose $p_{i, 2}$.
    Thus, each vertex encodes exactly one colorful hyperplane $h(v)$. 

    For each dimension $i\in [d]$, we define 
    \begin{align*}
        \orientationCube(v)_i := \begin{cases}
            0 & \text{if } \abs{h(v)^+ \cap P_i}=\alpha_i \\
            1 & otherwise.
        \end{cases}
    \end{align*}

    This construction can be done in polynomial time.
    Any sink of \orientationCube corresponds trivially to an $\alpha$-cut of the $P_i$.

    What is left to do is to prove that \orientationCube is USO. We first show that \orientationCube has a unique sink on the whole cube. This simply follows from \Cref{thm:alphaHam}, which says that there is a unique $\alpha$-cut.
    A subcube corresponds to looking at only a subset of the points. Since a family of points remains well-separated when removing points, the uniqueness of each possible $\alpha'$ cut also holds for every subcube.
    Thus, \orientationCube has a unique sink on every subcube, and is thus a USO.
\end{proof}

Now, we prove \Cref{lem:swsTwoHs2cubeUSO} for general $|P_i|$. For the case of two colors (thus resulting in a two-dimensional grid USO) this was already done by Felsner, Gärtner, and Tschirschnitz~\cite{felsner2005gridorientations}, where the setting was called ``one line and $n$ points''.

\begin{proof}[Proof of \Cref{lem:swsTwoHs2cubeUSO} for general $|P_i|$]
     We want to construct an orientation function $\orientationGrid$ on a $d$-dimensional grid $\Grid = (\abs{P_1}, \dots, \abs{P_d})$ with $V(\Grid) = [0, \dots,  \abs{P_1}] \times \dots \times [0, \dots,  \abs{P_d}]$, such that \orientationGrid is USO and the $\alpha$-cut of the point sets can be derived from its sink.
    Like in the proof of \Cref{lem:swsTwoHs2cubeUSO}, every vertex $v \in V(\Grid)$ corresponds to the colorful point set $\{p_{i, v_i+1}\;|\;i\in [d]\}$.

    Let $(v, w) \in E(\Grid)$ an edge of dimension $i$, i.e., $w\setminus v = \{p\} \subseteq P_i$. The orientation \orientationGrid is defined as:
    \begin{align*}
        \orientationGrid(v, w) := \begin{cases}
            0 & \text{if } \alpha_i = 1 \text{ and } p\not\in h(v)^+ \\
            0 & \text{if } \alpha_i = \abs{P_i} \text{ and } p\in h(v)^+\\
            1 & otherwise.
        \end{cases}
    \end{align*}

    This construction can be done in polynomial time.
    Any sink of \orientationGrid is trivially a colorful hyperplane that is the $\alpha$-cut, and vice versa. By \Cref{thm:alphaHam}, there is a unique sink in the whole grid. If we now remove any single point $p$ from some $P_i$, the resulting point set family is still well-separated. Furthermore, if we translate this family into a grid orientation as above, we get exactly the subgrid of $\sigma$ obtained by removing the direction corresponding to $p$. Since \Cref{thm:alphaHam} still applies to this subgrid, we get that this subgrid has a unique sink too. Since this argument can be repeated, all subgrids have a unique sink, and thus $\sigma$ describes a USO.
\end{proof}

The general position assumption is not really necessary in these reductions. If there are colorful hyperplanes containing more than $d$ points we can simply apply \Cref{thm:aHSwithTieBreaking} instead of \Cref{thm:alphaHam}, and break ties arbitrarily by always orienting the edges between the involved vertices downwards. This approach has also been used for degenerate \plcp{s}~\cite{borzechowski2023degeneracy}, and we will thus not describe it in detail.

\begin{remark}\label{rem:alphaHam2alphaGrid}
If we instead want to find an $\alpha'$-cut for arbitrary $\alpha'$, we can use the reduction from \Cref{lem:swsTwoHs2cubeUSO} with $\alpha=(1,\ldots,1)$, but instead of searching for a sink in the resulting grid USO $\Grid$, we need to find a vertex with $\alpha'_i -1$ outgoing edges in each dimension $i$.
    
    By \cite[Theorem 2.14]{gaertner2008grids}, this vertex is guaranteed to exist and to be unique. We call the problem of searching for such a vertex $\alpha$-\GridUSO. Note that $\alpha$-\GridUSO is not known to reduce to regular \GridUSO, nor is it known to lie in \Comp{Promise-UEOPL} (compared to $\alpha$-Ham Sandwich which does lie in \Comp{Promise-UEOPL}~\cite{chiu_et_al-ComplexityAlphaHamSandwich}).
\end{remark}

\subsection{\GridUSO to \CubeUSO}\label{sec:grid2cube}

In this section we show that every \GridUSO instance (\Grid, \orientationGrid) can be reduced to a \CubeUSO instance (C,\orientationCube) in polynomial time such that
given the global sink of \orientationCube, we can derive the global sink of \orientationGrid in polynomial time.

We first show how the reduction works for a one-dimensional grid, or in other words, a single simplex.

\subsubsection{One-Dimensional Warm-Up}
\label{sec:Simplex2Cube}

In this warm-up example, our given grid $\Grid$ is a single simplex $\simplex$, i.e., its vertex set is given by just one integer $n:=n_1$.
This simplex has $\partitionLength+1$ vertices $\{\vertex{0}, \dots, \vertex{\partitionLength}\}$.
We place $\simplex$ in a $\partitionLength$-dimensional hypercube $C_\simplex$ with $V(C_\simplex) = \{0,1\}^{\partitionLength}$.
Note that the simplex is an \partitionLength-regular graph, and so is the \partitionLength-dimensional hypercube.

\myparagraph{Embedding the vertices.}
Since the simplex has $\partitionLength+1$ vertices, and the hypercube has $2^\partitionLength$ vertices, we have a lot more vertices to play with in the cube. We embed the vertices of the grid in a star shape around the vertex $0^\dimensionB$: We place the vertex \vertex{0} at position $0^\dimensionB$, and each vertex $\vertex{v}$ for $v\in \{1, \dots, \partitionLength\}$ is placed at vertex $I_v$.
See \cref{fig:OrientingVertices} for an example.
We call these vertices of the cube (the unit vectors $I_v$ and $0^\dimensionB$) the \emph{grid-vertices}.
All other vertices are called \emph{non-grid-vertices}.
We assign each vertex~$J$ of the cube a \emph{color} $j\in \{0,\ldots, n\}$, where $j := max(\{0\}\cup \{h\in [\dimensionB] \;|\; J_h=1\})$. This naturally associates $J$ with a vertex $(j)$ in $\simplex$. Note that if $J$ itself is a grid-vertex, its associated vertex $(j)$ is the one that was placed into $J$. Further note that the assignment of colors is independent from the orientation of the simplex. \Cref{fig:OrientingVertices} shows the way colors are assigned to cube vertices.

\myparagraph{Orienting grid-vertices.}
The paths of length one between $0^n$ and $I_v$ encode the orientation of $(0)$ in the grid. Similarly, the paths of length two between grid-vertices $I_v$ and $I_w$ that avoid the vertex $0^n$ encode the orientation of the grid; if the edge between $(v)$ and $(w)$ is oriented towards $(w)$, we have a directed path from $I_v$ to $I_v\xor I_w$ to $I_w$. Note that by doing this, we have already completely oriented the grid-vertices. 
Formally, for $v,h \in [1, \dots, \dimensionB]$, we let
\begin{align}
\label{orientation:GridVertices}
\orientationCube(I_v)_h := 
\begin{cases}
\orientationGrid(\vertex{v}, \vertex{h}) & \text{if } h \neq v, \\
\orientationGrid(\vertex{v}, \vertex{0}) &  \text{otherwise,}
\end{cases}
\quad \quad
\orientationCube(0^\dimensionB)_h :=\orientationGrid(\vertex{0}, \vertex{h}).
\end{align}

See \cref{fig:OrientingVertices} for an example. We can see that if any grid-vertex in the cube is a sink, its corresponding vertex in the simplex is a sink too.

\myparagraph{Orienting non-grid-vertices.}
The only thing that is left to do is to extend this orientation to a complete orientation of the cube that is a valid USO.
We orient the rest of the cube according to the colors of the non-grid-vertices. Firstly, edges between vertices of two different colors are oriented as the corresponding edge in the grid; if in the grid the edge between $(j)$ and $(k)$ is oriented towards $(k)$, any edge in the cube between vertices $J$ and $K$ of colors $j$ and $k$ is oriented towards $K$.

Secondly, we wish to orient edges between vertices of the same color. We first observe that the set of vertices with the same color $j\in [n]$ actually forms a $(j-1)$-dimensional subcube. We now orient the edges within this subcube such that this subcube forms a \emph{uniform} orientation. This means that within this subcube, for every dimension $i$, all the edges of dimension $i$ are oriented the same way. Since the unique grid-vertex of every color is already completely oriented, this uniquely defines the orientation of the whole cube.

To formalize this, we introduce a \emph{secondary color} for every vertex: For vertex $J \in V(C_\simplex)$, the secondary color $j^*$ of $J$ is defined as
 $j^* := max(\{0\}\cup \{h\in [n]\setminus \{j\} \;|\; J_h=1\})$.
For all vertices $J$ except $0^\dimensionA$, $j^*$ describes the color of the vertex neighboring $J$ in dimension $j$.
Using this additional notation, we can now define our complete orientation.
For each dimension $h \in [1, \dots, \dimensionB]$ and each vertex $J\in V(C_\simplex)$, we orient $J$ as follows:
\begin{align}
\label{orientation:nonGridVerticesVersion1}
\orientationCube(J)_h :=\begin{cases}
\orientationGrid(\vertex{j}, \vertex{h}) \xor J_h & \text{if } h<j \\ 
\orientationGrid(\vertex{j}, \vertex{h})  & \text{if } h>j\\ 
\orientationGrid(\vertex{j}, \vertex{j^*}) & \text{if } h = j\\
\end{cases}
\end{align}

The first case describes an edge towards a vertex of the same color, the second and third cases describe edges to differently colored vertices.
See \Cref{fig:OrientingVertices} for an example.
Note that \Cref{orientation:nonGridVerticesVersion1} is consistent with \Cref{orientation:GridVertices}, and we thus do not need to distinguish between grid-vertices and non-grid-vertices.

\begin{figure}[h!]
    \centering
    \begin{tikzpicture}[scale=0.8, roundnode/.style={circle, draw=black, thick, minimum size=4mm}]


\node[roundnode, fill=yellow] (0) at (-7, 2) {$\vertex{0}$};
\node[roundnode, fill=green!40] (1) at (-7, 0.5) {$\vertex{1}$};
\node[roundnode, fill=red!60] (2) at (-7, -1) {$\vertex{2}$};
\node[roundnode, fill=blue!30] (3) at (-7, -2.5) {$\vertex{3}$};
\node[roundnode, fill=orange!60] (4) at (-7, -4) {$\vertex{4}$};

\begin{scope}[very thick,decoration={markings,mark=at position 0.75 with {\arrow{>}}}] 
\draw [postaction={decorate}] (1) -- (0);
\draw [postaction={decorate}] (1) -- (2);
\draw [postaction={decorate}] (3) -- (2);
\draw [postaction={decorate}, blue] (3) -- (4);
\end{scope}
\begin{scope}[very thick,decoration={markings,mark=at position 0.5 with {\arrow{>}}}] 
\path[-] (0) edge[bend right=40, postaction={decorate}] (2);
\path[-] (1) edge[bend right=40, postaction={decorate}] (3);

\path[-] (3) edge[bend left=50, postaction={decorate}] (0);
\path[-] (1) edge[bend right=50, postaction={decorate}] (4);

\path[-] (0) edge[bend right=60, postaction={decorate}] (4);
\path[-] (4) edge[bend left=40, postaction={decorate}, red] (2);

\end{scope}

\node[roundnode, fill=orange!60] (0001) at (0, -3) {$I_4$};
\node[roundnode, fill=orange!60] (0011) at (1, -2) {};
\node[roundnode, fill=orange!60] (0101) at (0, -1) {};
\node[roundnode, fill=orange!60] (0111) at (1, 0) {};
\node[roundnode, fill=orange!60] (1001) at (2, -3){};
\node[roundnode, fill=orange!60] (1011) at (3, -2) {};
\node[roundnode, fill=orange!60] (1101) at (2, -1) {};
\node[roundnode, fill=orange!60] (1111) at (3, 0) {};

\node[roundnode, fill=yellow] (0000) at (-4, -4.5) {$0^{n_i}$};
\node[roundnode, fill=blue!30] (0010) at (-1, -1.5) {$I_3$};
\node[roundnode, fill=red!60] (0100) at (-4, 0.5) {$I_2$};
\node[roundnode, fill=blue!30] (0110) at (-1, 2.5) {};
\node[roundnode, fill=green!40] (1000) at (5, -4.5) {$I_1$};
\node[roundnode, fill=blue!30] (1010) at (7, -1.5) {};
\node[roundnode, fill=red!60] (1100) at (5, 0.5) {};
\node[roundnode, fill=blue!30] (1110) at (7, 2.5) {};


\begin{scope}[every edge/.style={draw=black,very thick}]

\path [<-] (1111) edge[blue] (1110);
\path[->] (0110) edge[blue] (0111);
\path[->] (1010) edge[blue] (1011);

\path [<-] (0000) edge (1000);
\path [->] (0000) edge (0100);
\path [<-] (0000) edge (0010);
\path [->] (0000) edge (0001);

\path[->] (1000) edge (1100);
\path[->] (1000) edge (1010);
\path[->] (1000) edge (1001);

\path[<-] (0100) edge (1100);
\path[<-] (0100) edge (0110);

\path[<-] (0010) edge (1010);
\path[->] (0010) edge (0110);
\path[->] (0010) edge[blue] (0011);

\path[<-] (0001) edge (1001);
\path[<-] (0100) edge[red] (0101);
\path[->] (0001) edge[red] (0101);
\path[<-] (0001) edge (0011);

\path[<-] (1100) edge (1110);
\path[<-] (1100) edge (1101);

\path[->] (1010) edge (1110);

\path[->] (1001) edge (1101);
\path[<-] (1001) edge (1011);

\path[<-] (0110) edge (1110);

\path[<-] (0101) edge (1101);
\path[<-] (0101) edge (0111);

\path[<-] (0011) edge (1011);
\path[->] (0011) edge (0111);

\path [->] (1111) edge (1101);
\path [<-] (1111) edge (1011);
\path [->] (1111) edge (0111);
\end{scope}
\end{tikzpicture}
    \caption{Example of placing grid-vertices (large) in the cube, coloring the non-grid-vertices (small) and orienting the edges. Each edge of the grid becomes a path of length two between two grid-vertices, see the red edges. Edges between vertices of different colors are oriented the same as in the grid, see the blue edges. Each subcube of the same color is oriented in a uniform way, see for example the orange subcube.}
    \label{fig:OrientingVertices}
\end{figure}
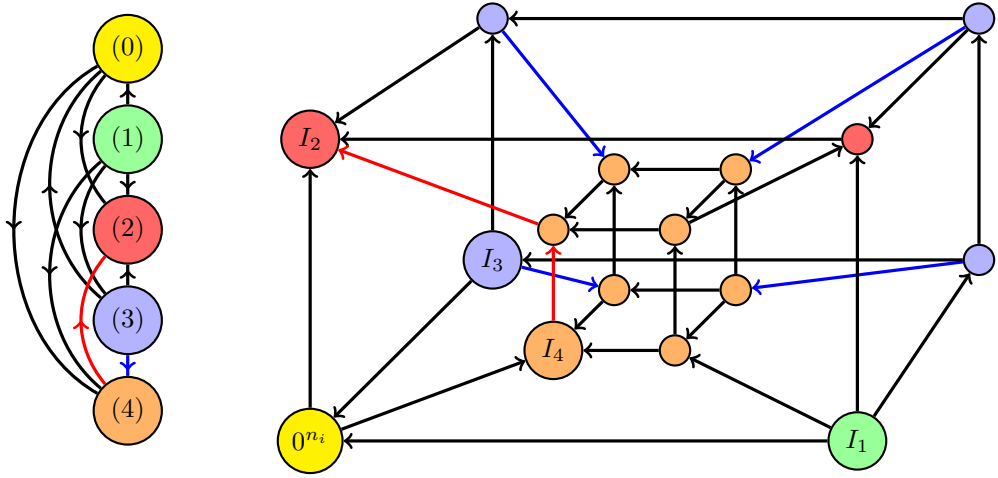

We refrain from proving correctness of the reduction at this point, since we will prove it in more generality later.

\subsubsection{General Reduction}
In full generality, we are given a \dimensionA-dimensional grid $\Grid = (\partitionLength_1, \dots, \partitionLength_\dimensionA)$ and wish to turn it into a hypercube of dimension $n := \sum\partitionLength_i$. Note that again \Grid and $C$ both have the same regularity: both are $n$-regular.

Intuitively, we apply the one-dimensional construction described in \cref{sec:Simplex2Cube} to every dimension of the grid at once.
For every dimension of the grid spanned by $n_i+1$ directions, we assign a block of $n_i$ dimensions in $C$. For simplicity, we index into the dimensions of the hypercube using double indices: for every bitstring $J \in \{0,1\}^{n}$, we write $J_{i, j}$ (where $j \in [\partitionLength_i]$) as a shorthand for~$J_{j+ \sum_{l\in [i-1]} \partitionLength_l}$. 

We extend our notion of colors from the one-dimensional warm-up: Each vertex $J\in V(C)$ is assigned a corresponding vertex in the grid $\Grid$ by simply defining a color $j_i$ for every grid dimension $i\in [d]$,
\begin{align}
j_i := max (\{0\} \cup \{ h \;\vert\;h \in [{\partitionLength_i}],  J_{i,h} = 1 \} ).\label{def:jPlus}
\end{align}

The tuple of colors $(j_1,\ldots,j_d)$ is the vertex in the grid that we associate with $J$. 
When we orient $J$, we orient each block $i$ of dimensions of $C$ independently. That means to compute $O(J)_{i,h}$, we simply compute $O(J_{i,\cdot})_h$ according to the one-dimensional warm-up, with the input simplex being the simplex in dimension $i$ in $\Grid$ that contains the vertex $(j_1,\ldots,j_d)$.

To formalize this construction, we again need to define a secondary color $j_i^*$ for each vertex $J\in V(C)$ and each grid dimension $i\in[d]$,
\begin{align}
    j_i^* := max ( \{0\} \cup \{ h \;|\; h \in [{\partitionLength_i}] \setminus \{j_i\},  J_{i,h} = 1 \} ). \label{def:jMinus}
\end{align}

The orientation $O$ of the cube $C$ is now defined as follows:
\begin{align}
\label{orientation:general}
\orientationCube(J)_{i,h} := \begin{cases}
\orientationGrid(\vertex{j_1, \dots,j_i, \dots, j_\dimensionA}, \vertex{j_1, \dots, h, \dots j_\dimensionA}) \xor J_{i,h} & \text{for } h < j_i,\\
\orientationGrid(\vertex{j_1, \dots,j_i, \dots, j_\dimensionA}, \vertex{j_1, \dots,  h, \dots j_\dimensionA}) & \text{for } h > j_i,\\
\orientationGrid(\vertex{j_1, \dots,j_i, \dots, j_\dimensionA}, \vertex{j_1, \dots, j_i^*, \dots j_\dimensionA}) & \text{for } h = j_i.
\end{cases}
\end{align}

\Cref{fig:grid2cubeExample} shows an example of the construction described by \Cref{def:jPlus,def:jMinus,orientation:general} where a $2$-dimensional $5\times 5$ grid is converted to an $8$-dimensional cube.

\newcommand{\CubeAsVertex}[2]{
\begin{tikzpicture}[scale=0.12, 
roundnode/.style={circle, draw=black, minimum size=5pt, inner sep=0pt},
starnode/.style={star, star points=4, star point ratio=0.4, draw=black, minimum size=3pt, inner sep=0pt},
rectnode/.style={rectangle, draw=black, minimum size=5pt, inner sep=0pt},
diamnode/.style={diamond, draw=black, minimum size=6pt, inner sep=0pt},
pentnode/.style={regular polygon, regular polygon sides=5, draw=black, minimum size=6pt, inner sep=0pt}
]

\node[#2, fill=orange!60] (0001#1) at (0, -3) {};
\node[#2, fill=orange!60] (0011#1) at (1, -2) {};
\node[#2, fill=orange!60] (0101#1) at (0, -1) {};
\node[#2, fill=orange!60] (0111#1) at (1, 0) {};
\node[#2, fill=orange!60] (1001#1) at (2, -3){};
\node[#2, fill=orange!60] (1011#1) at (3, -2) {};
\node[#2, fill=orange!60] (1101#1) at (2, -1) {};
\node[#2, fill=orange!60] (1111#1) at (3, 0) {};

\node[#2, fill=yellow, scale=2] (0000#1) at (-4, -4.5) {};
\node[#2, fill=blue!30] (0010#1) at (-1, -1.5) {};
\node[#2, fill=red!60] (0100#1) at (-4, 0.5) {};
\node[#2, fill=blue!30] (0110#1) at (-1, 2.5) {};
\node[#2, fill=green!40] (1000#1) at (5, -4.5) {};
\node[#2, fill=blue!30] (1010#1) at (7, -1.5) {};
\node[#2, fill=red!60] (1100#1) at (5, 0.5) {};
\node[#2, fill=blue!30] (1110#1) at (7, 2.5) {};


\begin{scope}[every edge/.style={draw=black,thick}]
\path [-] (0000#1) edge (1000#1);
\path [-] (0000#1) edge (0100#1);
\path [-] (0000#1) edge (0010#1);
\path [-] (0000#1) edge (0001#1);

\path[-] (1000#1) edge (1100#1);
\path[-] (1000#1) edge (1010#1);
\path[-] (1000#1) edge (1001#1);

\path[-] (0100#1) edge (1100#1);
\path[-] (0100#1) edge (0110#1);
\path[-] (0100#1) edge (0101#1);

\path[-] (0010#1) edge (1010#1);
\path[-] (0010#1) edge (0110#1);
\path[-] (0010#1) edge (0011#1);

\path[-] (0001#1) edge (1001#1);
\path[-] (0001#1) edge (0101#1);
\path[-] (0001#1) edge (0011#1);

\path[-] (1100#1) edge (1110#1);
\path[-] (1100#1) edge (1101#1);

\path[-] (1010#1) edge (1110#1);
\path[-] (1010#1) edge (1011#1);

\path[-] (1001#1) edge (1101#1);
\path[-] (1001#1) edge (1011#1);

\path[-] (0110#1) edge (1110#1);
\path[-] (0110#1) edge (0111#1);

\path[-] (0101#1) edge (1101#1);
\path[-] (0101#1) edge (0111#1);

\path[-] (0011#1) edge (1011#1);
\path[-] (0011#1) edge (0111#1);

\path [-] (1111#1) edge (1110#1);
\path [-] (1111#1) edge (1101#1);
\path [-] (1111#1) edge (1011#1);
\path [-] (1111#1) edge (0111#1);
\end{scope}
\end{tikzpicture}
}

\newcommand{\CubeAsVertexOrigin}[1]{
\begin{tikzpicture}[scale=0.3, roundnode/.style={circle, draw=black, minimum size=10pt, inner sep=0pt}]

\node[roundnode, fill=orange!60] (0001#1) at (0, -3) {4};
\node[roundnode, fill=orange!60] (0011#1) at (1, -2) {};
\node[roundnode, fill=orange!60] (0101#1) at (0, -1) {};
\node[roundnode, fill=orange!60] (0111#1) at (1, 0) {};
\node[roundnode, fill=orange!60] (1001#1) at (2, -3){};
\node[roundnode, fill=orange!60] (1011#1) at (3, -2) {};
\node[roundnode, fill=orange!60] (1101#1) at (2, -1) {};
\node[roundnode, fill=orange!60] (1111#1) at (3, 0) {};

\node[roundnode, fill=yellow] (0000#1) at (-4, -4.5) {0};
\node[roundnode, fill=blue!30] (0010#1) at (-1, -1.5) {3};
\node[roundnode, fill=red!60] (0100#1) at (-4, 0.5) {2};
\node[roundnode, fill=blue!30] (0110#1) at (-1, 2.5) {};
\node[roundnode, fill=green!40] (1000#1) at (5, -4.5) {1};
\node[roundnode, fill=blue!30] (1010#1) at (7, -1.5) {};
\node[roundnode, fill=red!60] (1100#1) at (5, 0.5) {};
\node[roundnode, fill=blue!30] (1110#1) at (7, 2.5) {};


\begin{scope}[every edge/.style={draw=black,thick}]
\path [-] (0000#1) edge (1000#1);
\path [-] (0000#1) edge (0100#1);
\path [-] (0000#1) edge (0010#1);
\path [-] (0000#1) edge (0001#1);

\path[-] (1000#1) edge (1100#1);
\path[-] (1000#1) edge (1010#1);
\path[-] (1000#1) edge (1001#1);

\path[-] (0100#1) edge (1100#1);
\path[-] (0100#1) edge (0110#1);
\path[-] (0100#1) edge (0101#1);

\path[-] (0010#1) edge (1010#1);
\path[-] (0010#1) edge (0110#1);
\path[-] (0010#1) edge (0011#1);

\path[-] (0001#1) edge (1001#1);
\path[-] (0001#1) edge (0101#1);
\path[-] (0001#1) edge (0011#1);

\path[-] (1100#1) edge (1110#1);
\path[-] (1100#1) edge (1101#1);

\path[-] (1010#1) edge (1110#1);
\path[-] (1010#1) edge (1011#1);

\path[-] (1001#1) edge (1101#1);
\path[-] (1001#1) edge (1011#1);

\path[-] (0110#1) edge (1110#1);
\path[-] (0110#1) edge (0111#1);

\path[-] (0101#1) edge (1101#1);
\path[-] (0101#1) edge (0111#1);

\path[-] (0011#1) edge (1011#1);
\path[-] (0011#1) edge (0111#1);

\path [-] (1111#1) edge (1110#1);
\path [-] (1111#1) edge (1101#1);
\path [-] (1111#1) edge (1011#1);
\path [-] (1111#1) edge (0111#1);
\end{scope}
\end{tikzpicture}
}

\begin{figure}
    \centering
\begin{tikzpicture}[
scale=1.1,
vertexnode/.style={inner sep=0pt},
roundnode/.style={circle, draw=black, minimum size=10pt, inner sep=0pt},
starnode/.style={star, star points=4, star point ratio=0.4, draw=black, minimum size=6pt, inner sep=0pt},
rectnode/.style={rectangle, draw=black, minimum size=10pt, inner sep=0pt},
diamnode/.style={diamond, draw=black, minimum size=11pt, inner sep=0pt},
pentnode/.style={regular polygon, regular polygon sides=5, draw=black, minimum size=11pt, inner sep=0pt}
]


\foreach \i in {0,...,4}{\node at (\i, 4) {\i};}
\foreach \i in {0,...,4}{\node at (-1, 5+\i) {\i};}

\node[roundnode, fill=yellow] (50) at (0, 5) {};
\node[starnode, fill=yellow] (51) at (1, 5) {};
\node[rectnode, fill=yellow] (52) at (2, 5) {};
\node[diamnode, fill=yellow] (53) at (3, 5) {};
\node[pentnode, fill=yellow] (54) at (4, 5) {};
\node[roundnode, fill=green!40] (60) at (0, 6) {};
\node[starnode, fill=green!40] (61) at (1, 6) {};
\node[rectnode, fill=green!40] (62) at (2, 6) {};
\node[diamnode, fill=green!40] (63) at (3, 6) {};
\node[pentnode, fill=green!40] (64) at (4, 6) {};
\foreach \i in {0,...,4}{\draw (5\i) -- (6\i);}
\node[roundnode, fill=red!60] (70) at (0, 7) {};
\node[starnode, fill=red!60] (71) at (1, 7) {};
\node[rectnode, fill=red!60] (72) at (2, 7) {};
\node[diamnode, fill=red!60] (73) at (3, 7) {};
\node[pentnode, fill=red!60] (74) at (4, 7) {};
\foreach \i in {0,...,4}{\draw (7\i) -- (6\i);}
\node[roundnode, fill=blue!30] (80) at (0, 8) {};
\node[starnode, fill=blue!30] (81) at (1, 8) {};
\node[rectnode, fill=blue!30] (82) at (2, 8) {};
\node[diamnode, fill=blue!30] (83) at (3, 8) {};
\node[pentnode, fill=blue!30] (84) at (4, 8) {};
\foreach \i in {0,...,4}{\draw (7\i) -- (8\i);}
\node[roundnode, fill=orange!60] (90) at (0, 9) {};
\node[starnode, fill=orange!60] (91) at (1, 9) {};
\node[rectnode, fill=orange!60] (92) at (2, 9) {};
\node[diamnode, fill=orange!60] (93) at (3, 9) {};
\node[pentnode, fill=orange!60] (94) at (4, 9) {};
\foreach \i in {0,...,4}{\draw (9\i) -- (8\i);}

\foreach \i in {5,...,9}{
\draw (\i1) -- (\i0);
\draw (\i1) -- (\i2);
\draw (\i3) -- (\i2);
\draw (\i3) -- (\i4);
}
\path[-] (90) edge[bend left=40] (92);
\path[-] (91) edge[bend left=40] (93);
\path[-] (92) edge[bend left=40] (94);
\path[-] (90) edge[bend left=50] (93);
\path[-] (91) edge[bend left=50] (94);
\path[-] (90) edge[bend left=60] (94);

\path[-] (54) edge[bend right=40] (74);
\path[-] (64) edge[bend right=40] (84);
\path[-] (74) edge[bend right=40] (94);
\path[-] (54) edge[bend right=50] (84);
\path[-] (64) edge[bend right=50] (94);
\path[-] (54) edge[bend right=60] (94);

\node[vertexnode] (0001) at (0, -3) {\CubeAsVertex{0001}{pentnode}};
\node[vertexnode] (0011) at (1, -2) {\CubeAsVertex{0011}{pentnode}};
\node[vertexnode] (0101) at (0, -1) {\CubeAsVertex{0101}{pentnode}};
\node[vertexnode] (0111) at (1, 0) {\CubeAsVertex{0111}{pentnode}};
\node[vertexnode] (1001) at (2, -3){\CubeAsVertex{1001}{pentnode}};
\node[vertexnode] (1011) at (3, -2) {\CubeAsVertex{1011}{pentnode}};
\node[vertexnode] (1101) at (2, -1) {\CubeAsVertex{1101}{pentnode}};
\node[vertexnode] (1111) at (3, 0) {\CubeAsVertex{1111}{pentnode}};

\node[vertexnode] (0000) at (-4, -4.5) {\CubeAsVertex{0000}{roundnode}};
\node[vertexnode] (0010) at (-1, -1.5) {\CubeAsVertex{0010}{diamnode}};
\node[vertexnode] (0100) at (-4, 0.5) {\CubeAsVertex{0100}{rectnode}};
\node[vertexnode] (0110) at (-1, 2.5) {\CubeAsVertex{0110}{diamnode}};
\node[vertexnode] (1000) at (5, -4.5) {\CubeAsVertex{1000}{starnode}};
\node[vertexnode] (1010) at (7, -1.5) {\CubeAsVertex{1010}{diamnode}};
\node[vertexnode] (1100) at (5, 0.5) {\CubeAsVertex{1100}{rectnode}};
\node[vertexnode] (1110) at (7, 2.5) {\CubeAsVertex{1110}{diamnode}};


\begin{scope}[every edge/.style={draw=black,very thick}]
\path [-] (0000) edge (1000);
\path [-] (0000) edge (0100);
\path [-] (0000) edge (0010);
\path [-] (0000) edge (0001);

\path[-] (1000) edge (1100);
\path[-] (1000) edge (1010);
\path[-] (1000) edge (1001);

\path[-] (0100) edge (1100);
\path[-] (0100) edge (0110);
\path[-] (0100) edge (0101);

\path[-] (0010) edge (1010);
\path[-] (0010) edge (0110);
\path[-] (0010) edge (0011);

\path[-] (0001) edge (1001);
\path[-] (0001) edge (0101);
\path[-] (0001) edge (0011);

\path[-] (1100) edge (1110);
\path[-] (1100) edge (1101);

\path[-] (1010) edge (1110);
\path[-] (1010) edge (1011);

\path[-] (1001) edge (1101);
\path[-] (1001) edge (1011);

\path[-] (0110) edge (1110);
\path[-] (0110) edge (0111);

\path[-] (0101) edge (1101);
\path[-] (0101) edge (0111);

\path[-] (0011) edge (1011);
\path[-] (0011) edge (0111);

\path [-] (1111) edge (1110);
\path [-] (1111) edge (1101);
\path [-] (1111) edge (1011);
\path [-] (1111) edge (0111);
\end{scope}
\end{tikzpicture}
 \caption{Example of a $5\times 5$ grid (top, not all edges drawn) integrated in a 8-dimensional hypercube (bottom). The associated grid vertex of every vertex of the cube is indicated using the color and the shape of the vertex.}
    \label{fig:grid2cubeExample}
\end{figure}
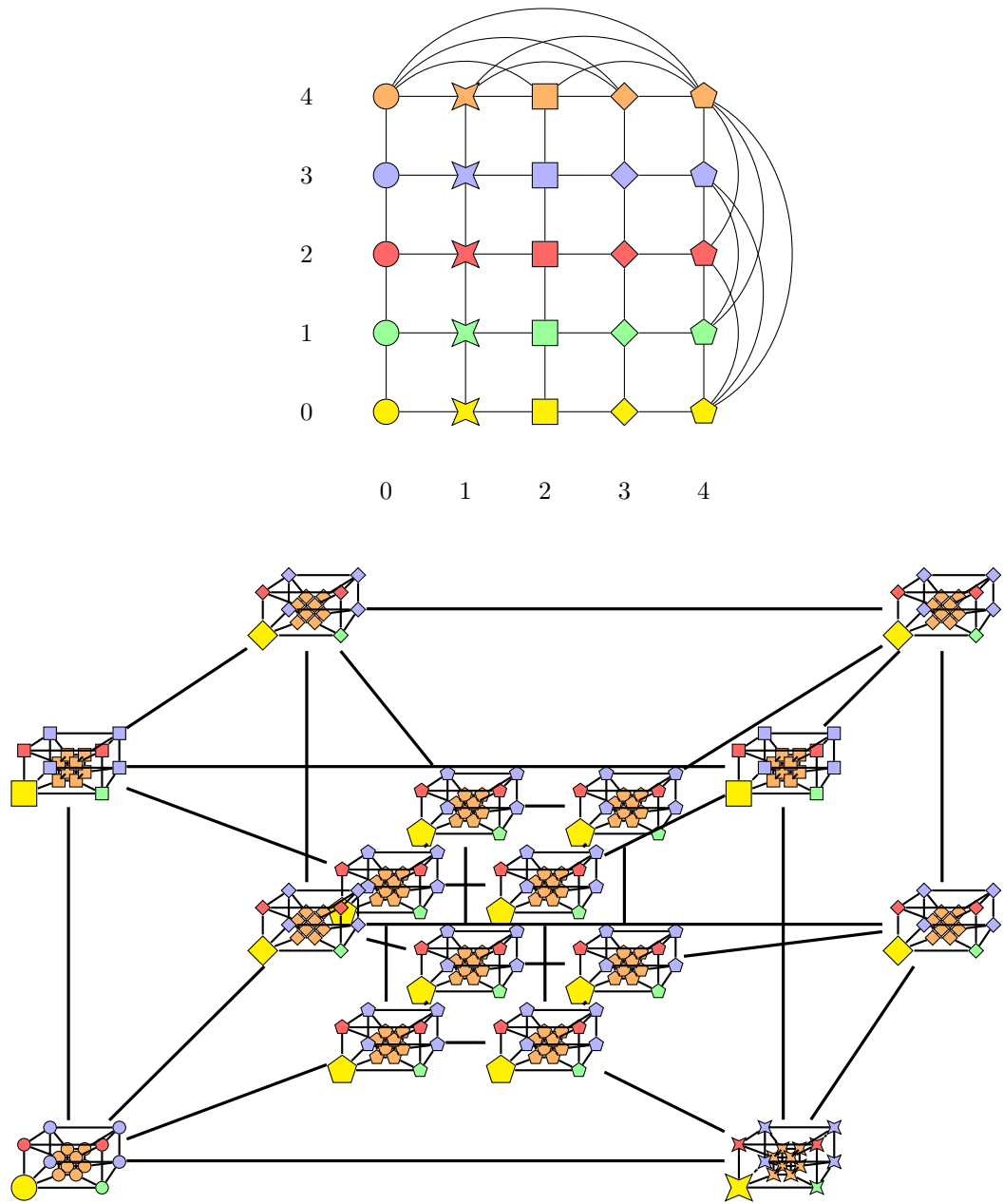

\begin{theorem}
\label{thm:generalizationIsUSO}
If \orientationGrid is a unique sink orientation, then the orientation constructed by \Cref{orientation:general} is a unique sink orientation. 
A circuit computing \orientationCube can be computed in polynomial time.
Given the sink of \orientationCube, the sink of \orientationGrid can be found in polynomial time.
\end{theorem}

To prove this theorem, we first need some additional ingredients. We use the following characterization of cube USOs:
\begin{lemma}[\szabo{}-Welzl Condition \cite{szabo2001usos}]\label[lemma]{lem:szabowelzl}
An orientation \orientationCube of a \dimension-dimensional cube $C$ is USO if and only if for all pairs of distinct vertices $J, K \in V(C)$, we have
\[\exists i\in J \xor K:\; \orientationCube(J)_i \not= \orientationCube(K)_i.\]
\end{lemma}

We also need some more definitions for the proof. A vertex $J$ is in the \emph{upper $i$-facet} of $C$ if $J_i = 1$ and in the \emph{lower $i$-facet} otherwise.
Two vertices $J, K \in V(C)$ are called \emph{antipodal} iff $J_i \neq K_i$ for all $i \in [\dimension]$.

To prove \orientationCube to be a USO we rely on proofs by contradiction. While every orientation that is not a USO fails the \SWC for some pair of vertices, this is sometimes not enough to lead to a contradiction. We thus need something stronger: Every orientation that is not a USO contains a minimal face that is not a USO, i.e., a face that does not have a unique sink, but all of its proper faces do. There is a surprising amount of structure in these minimal faces, and they have been studied extensively by Bosshard and Gärtner \cite{gaertner2006lpuso} who named them \emph{pseudo USOs}:
\begin{definition}
A \emph{pseudo unique sink orientation} (PUSO) of $C$ is an orientation \orientationCube that does not have a unique sink, but every induced subcube $f \neq C$ has a unique sink.
\end{definition}
The main property of pseudo USOs that we will use in this section is that every pair of antipodal vertices fails the \SWC:
\begin{lemma}[{\cite[Corollary 6 and Lemma 8]{bosshard2017pseudo}}]\label[lemma]{lem:PUSO}
Let \orientationCube be a PUSO. Then, for every pair of antipodal vertices $J$ and $K$ it holds that $O(J)=O(K)$.
\end{lemma}

We are now ready to prove \Cref{thm:generalizationIsUSO}.
\begin{proof}[Proof of \Cref{thm:generalizationIsUSO}]
It is clear that using \Cref{orientation:general}, a circuit computing $O$ can be built from a circuit computing $\sigma$ in polynomial time.

Next we show that \orientationGrid being a USO implies that \orientationCube as constructed by \Cref{def:jPlus,def:jMinus,orientation:general} is a USO. We prove this by contraposition, and thus assume that \orientationCube is not a USO. 
Then $O$ contains a minimal face that is not a USO, i.e., a PUSO. Let $J, K \in V(C)$ be the pair of antipodal vertices in that PUSO that are its minimum and maximum vertices, i.e., $J\subset K$. By \Cref{lem:PUSO} we know that $\orientationCube(J)_i = \orientationCube(K)_i$ for all $i \in J \xor K$.
We will show that this implies that \orientationGrid is not a USO.


Let $M \subseteq [\dimensionA]$ be the set of grid-dimensions in which the grid-vertices associated with $J$ and $K$ have a different color, i.e., $m\in M$ if $j_m\neq k_m$. Then, we have
\begin{align*}
\forall m \in M, \quad &j_m < k_m \text{ and } (m, k_m) \in J \xor K, \text{ since $J\subset K$, and we also have}\\
\forall  m \in [\dimensionA]\setminus M, \quad &j_m = k_m \text{ and } (m, k_m) \notin J \xor K.
\end{align*}

If $M = \emptyset$, then $J$ and $K$ are both associated with the same vertex of the grid.
Since for every $m\in [d]$ and every $h\geq k_m$ we have $J_{m,h}=K_{m,h}$ and $J\neq K$, there must exist an $m\in [\dimensionA], h<k_m$ such that $(m,h)\in J\xor K$.
We then have:
\begin{itemize}
    \item $ \orientationCube(J)_{m, h} = \orientationGrid(\vertex{j_1, \dots,j_m, \dots,  j_\dimensionA}, \vertex{j_1, \dots, h, \dots  j_\dimensionA} ) \xor J_{m, h}.$
    \item $\orientationCube(K)_{m, h} = \orientationGrid(\vertex{j_1, \dots,j_m, \dots, j_\dimensionA}, \vertex{j_1, \dots, h, \dots  j_\dimensionA} ) \xor K_{m, h}.$
\end{itemize}
Since $J_{m, h} = 0$ and $K_{m, h} = 1$ we must have $\orientationCube(J)_{m, h} \neq \orientationCube(K)_{m, h} $, which is a contradiction to the assumption that $O(J)_i=O(K)_i$ for all $i\in J\xor K$. Thus we know that if \orientationCube is not USO, then $M \neq \emptyset$.

If $M\neq \emptyset$, there are two different vertices of the grid associated with $J$ and $K$: \vertex{j_1, \dots, j_\dimensionA} and \vertex{k_1, \dots, k_\dimensionA}, respectively. These vertices span a subcube in the grid \Grid in exactly the dimensions of $M$, involving the directions $j_m$ and $k_m$ for all $m\in M$.


Let $M'\subseteq M$ be the set $\{m \;|\; m\in M \text{ and } (m,k_m^*)\in J\xor K\}$. For all $m\in M'$ we know that $k_m >k_m^* > j_m$, since $J\subset K$. We now consider the vertices $J'$ and $K'$ of the cube obtained by walking from $J$ and $K$ in the dimensions $(m,k_m^*)$ for all $m\in M'$, i.e., 
$$J'=J\xor\{(m,k_m^*) \;|\; m\in M'\} \text{ and } K'=K\xor\{(m,k_m^*)\;|\; m\in M'\}.$$

We know that the vertex of the grid associated with $K'$ is still \vertex{k_1, \dots, k_\dimensionA}, the same as with~$K$. On the other hand, the vertex associated with $J'$ is \vertex{l_1,\ldots,l_d}, where $l_i=k_i^*$ if $i\in M'$ and $l_i=j_i$ otherwise. Furthermore, $J'$ and $K'$ still form a pair of antipodal vertices in the PUSO spanned by $J$ and $K$, since $J'\xor J=K'\xor K\subseteq J\xor K$. Thus, we know that $O(J')_i=O(K')_i$ for all $i\in J\xor K$.

Let us now compute the orientations of $O(J')$ and $O(K')$ in the dimensions $(m,k_m)$ for all $m\in M$. Note that for all such $m$, we have $l_m=k_m^*<k_m$:
For $(m, k_m^*) \in J \xor K$, i.e., for $m\in M'$ we already argued this above.
For $(m, k_m^*) \notin J \xor K$ and $k_m^*\neq 0$, we know that $J_{m, k_m^*} = K_{m, k_m^*} = J'_{m, k_m^*} = K'_{m, k_m^*}=1$, and thus $k_m^* = j_m=l_m$.
Finally, for $k_m^*=0$ we must have that $j_m=l_m=0$ too.

By \Cref{orientation:general} we now get  for all $m\in M$:
\begin{itemize}
\item $\orientationCube(J')_{(m,k_m)} = \orientationGrid(\vertex{l_1, \dots,l_m, \dots l_\dimensionA}, \vertex{l_1, \dots, k_m, \dots l_\dimensionA})$ by Case 2, since $l_m<k_m$.
\item $\orientationCube(K')_{(m,k_m)} = \orientationGrid(\vertex{k_1, \dots,k_m, \dots k_\dimensionA}, \vertex{k_1, \dots, k_m^*, \dots k_\dimensionA})$ by Case 3. Since $k_m^*=l_m$, this is equivalent to $\orientationCube(K')_{(m,k_m)} = \orientationGrid(\vertex{k_1, \dots,k_m, \dots k_\dimensionA}, \vertex{k_1, \dots, l_m, \dots k_\dimensionA})$.
\end{itemize}

Since we have that $O(J')_{(m,k_m)}=O(K')_{(m,k_m)}$ for all $m\in M$, we can see that the vertices of the grid $(l_1,\ldots,l_d)$ and $(k_1,\ldots,k_\dimensionA)$ have the same outmap in the subcube they span (which is exactly spanned by the dimensions $M$).
See \Cref{fig:Contradiction} for an example.
Thus, $\sigma$ cannot describe a USO.

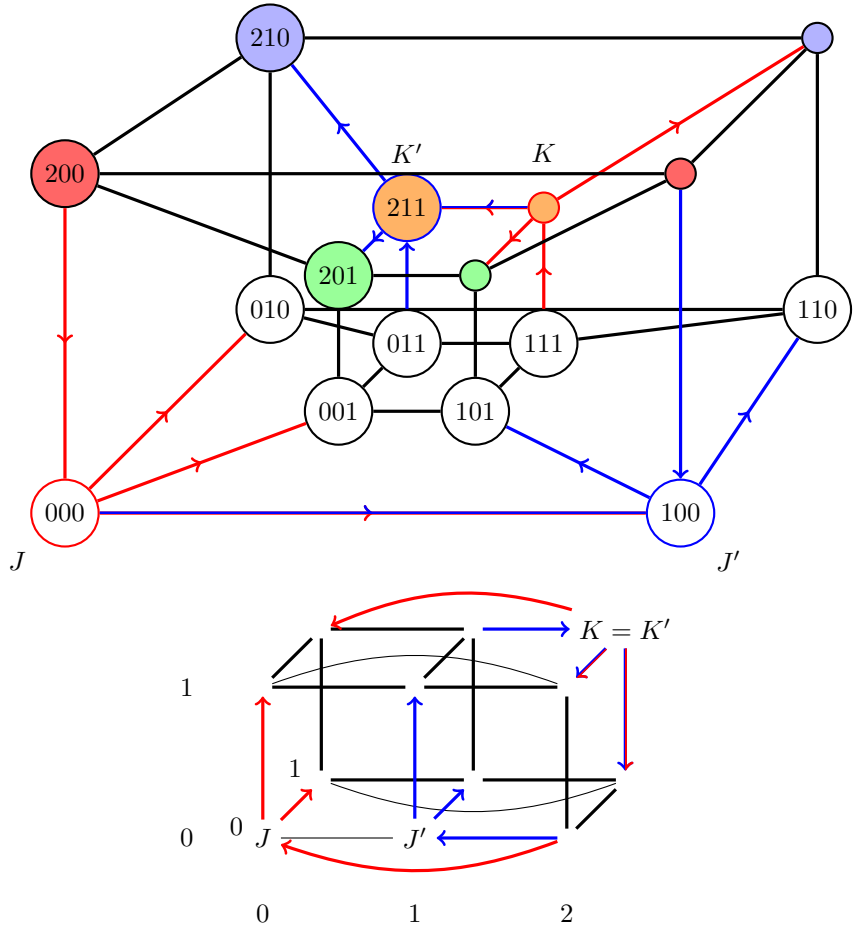
\begin{figure}[h!]
\centering
\begin{tikzpicture}[scale=0.9, roundnode/.style={circle, draw=black, thick, minimum size=4mm}]

\node[roundnode] (0001) at (0, -3) {$001$};
\node[roundnode] (0011) at (1, -2) {$011$};
\node[roundnode, fill=green!40] (0101) at (0, -1) {$201$};
\node[roundnode, fill=orange!60, draw=blue] (0111) at (1, 0) {$211$};
\node[roundnode] (1001) at (2, -3){$101$};
\node[roundnode] (1011) at (3, -2) {$111$};
\node[roundnode, fill=green!40] (1101) at (2, -1) {};
\node[roundnode, fill=orange!60, draw=red] (1111) at (3, 0) {};

\node[roundnode, draw=red] (0000) at (-4, -4.5) {$000$};
\node[roundnode] (0010) at (-1, -1.5) {$010$};
\node[roundnode, fill=red!60] (0100) at (-4, 0.5) {$200$};
\node[roundnode, fill=blue!30] (0110) at (-1, 2.5) {$210$};
\node[roundnode, draw=blue] (1000) at (5, -4.5) {$100$};
\node[roundnode] (1010) at (7, -1.5) {$110$};
\node[roundnode, fill=red!60] (1100) at (5, 0.5) {};
\node[roundnode, fill=blue!30] (1110) at (7, 2.5) {};

\node (J) at  (-4.7, -5.2) {$J$};
\node (J) at  (3, 0.8) {$K$};

\node (J) at  (5.7, -5.2) {$J'$};
\node (J) at  (1, 0.8) {$K'$};

\begin{scope}[very thick,decoration={markings,mark=at position 0.5 with {\arrow{>}}}] 

\draw [postaction={decorate}, red] (0000) -- (1000);
\begin{scope}
\clip (-4,-4.5) -- (5,-4.5) -- (5,-4) -- (-4,-4) -- cycle ;
\draw [postaction={decorate}, blue] (0000) -- (1000);
\end{scope}

\draw [postaction={decorate}, red] (0100) -- (0000);
\draw [postaction={decorate}, red] (0000) -- (0010);
\draw [postaction={decorate}, red] (0000) -- (0001);

\draw [postaction={decorate}, blue] (1000) -- (1010);
\draw [postaction={decorate}, blue] (1000) -- (1001);

\path[-] (0100) edge (0110);
\path[-] (0100) edge (0101);

\path[-] (0010) edge (1010);
\path[-] (0010) edge (0110);
\path[-] (0010) edge (0011);

\path[-] (0001) edge (1001);
\path[-] (0001) edge (0101);
\path[-] (0001) edge (0011);

\path[-] (1100) edge (1110);

\path[-] (1010) edge (1110);
\path[-] (1010) edge (1011);

\path[-] (1001) edge (1101);
\path[-] (1001) edge (1011);

\path[-] (0110) edge (1110);

\path[-] (0011) edge (1011);

\path[->, blue] (0011) edge (0111);
\draw [postaction={decorate}, blue] (0111) -- (0110);
\draw [postaction={decorate}, blue] (0111) -- (0101);

\draw [postaction={decorate}, red] (1111) -- (1110);
\draw [postaction={decorate}, red] (1111) -- (1101);
\draw [postaction={decorate}, red] (1011) -- (1111);

\draw [postaction={decorate}, red] (1111) -- (0111);
\begin{scope}
\clip (3,0) -- (1,0) -- (1,0.5) -- (3,0.5) -- cycle ;
\draw [postaction={decorate}, blue] (1111) -- (0111);
\end{scope}

\path[-] (0100) edge (1100);
\path[->, blue] (1100) edge (1000);
\path[-] (1100) edge (1101);
\path[-] (0101) edge (1101);

\end{scope}
\end{tikzpicture}
\begin{tikzpicture}
\node (000) at (0,0,0) {$J$};   
\node (010) at (0,2,0) {};
\node (110) at (2,2,0) {};
\node (100) at (2,0,0) {$J'$};
\node (001) at (0,0,-2) {};   
\node (011) at (0,2,-2) {};
\node (111) at (2,2,-2) {};
\node (101) at (2,0,-2) {};
\node (211) at (4,2,-2) {$K=K'$};
\node (201) at (4,0,-2) {};
\node (210) at (4,2,0) {};
\node (200) at (4,0,0) {};

\foreach \i in {0,...,2}{\node at (2*\i, -1, 0) {\i};}
\foreach \i in {0,1}{\node at (-1, 2*\i, 0) {\i};}
\foreach \i in {0,1}{\node at (-1.5, -1, - 2*\i - 3) {\i};}

\path[-] (001) edge[bend right=20] (201);
\path[->, very thick, red] (000) edge (010); 
\path[-] (000) edge (100);
\path[->, very thick, red] (000) edge (001);
\path[-] (010) edge[bend left=20] (210);
\path[-, very thick] (010) edge (110);
\path[-, very thick] (010) edge (011);
\path[<-, very thick, blue] (100) edge (200);
\path[->, very thick, blue] (100) edge (101);
\path[-, very thick] (001) edge (101);
\path[-, very thick] (001) edge (011);
\path[-, very thick] (111) edge (011);
\path[-, very thick] (111) edge (101);
\path[-, very thick] (111) edge (110);
\path[->, very thick, blue] (111) edge (211);
\path[<-, very thick, blue] (201) edge (211);
\begin{scope}
    \clip (4,2,-2) -- (4,0,-2) -- (4.5,0,-2) -- (4.5,2,-2) -- cycle ;
    \path[<-, very thick, red] (201) edge (211);
\end{scope}

\path[-, very thick] (201) edge (101);
\path[-, very thick] (201) edge (200);
\path[-, very thick] (210) edge (200);
\path[-, very thick] (210) edge (110);
\path[<-, very thick, blue] (210) edge (211);
\begin{scope}
    \clip (4,2,-2) -- (4,2,0) -- (4,1.5,0) -- (4,1.5,-2) -- cycle ;
    \path[<-, very thick, red] (210) edge (211);
\end{scope}

\path[<-, very thick, red] (000) edge[bend right=20] (200);
\path[<-, very thick, red] (011) edge[bend left=20] (211);
\path[->, very thick, blue] (100) edge (110);

\end{tikzpicture}
\caption{Example: The vertices $J\subset K$ are minimal and maximal vertices of a pseudo USO in the cube. Their analogues in the grid do not violate the \SWC. Thus, we look at $J'$ and $K'$, which are also antipodal vertices in the pseudo USO and thus also violate the condition in the cube. Their analogues in the grid now also violate the \SWC (see the blue edges).}
\label{fig:Contradiction}
\end{figure}

It only remains to prove that given the global sink of \orientationCube, we can derive the global sink of \orientationGrid in polynomial time. To do this, we simply show that for the global sink of \orientationCube, its associated vertex of $\orientationGrid$ must be the global sink. Since we already proved that \orientationCube is a USO, it is enough to show that given the unique sink $j_1,\ldots,j_d$ of $\orientationGrid$, its associated grid-vertex $J$ in \orientationCube must be \emph{a} global sink. Since the colors of $J$ are $j_1,\ldots,j_d$, and $J_{i,h}=0$ for $h<j_i$, we must have that $O(J)_{i,h}=0$ for all $i,h$ by \Cref{orientation:general}, and thus $J$ is a (and thus the unique) global sink, proving the theorem.
\end{proof}

\section{Open Questions}

\myparagraph{Total search problem versions.} 
All reductions we provided in this paper are between the promise problem versions of the involved problems. It may be interesting to also find reductions between the total search problem versions.

\myparagraph{Missing reductions.}
We were unable to show that finding an $\alpha$-cut for arbitrary $\alpha$ is not more difficult than finding it for $\alpha_i\in\{1,|P_i|\}$. There might thus be a difference in the computational complexity of \aHS and \swsHS.
Similarly, we also do not know whether $\alpha$-\GridUSO (searching for a vertex with a specified number of outgoing edges per dimension) is not more difficult than \GridUSO (searching for a sink).
Note that in the case of \aHS, it is at least known that the problem is contained in \Comp{UEOPL}. This is not known for $\alpha$-\GridUSO.

\myparagraph{Semantics of the \GridUSO to \CubeUSO reduction.}
On the levels of USOs, we do not know the exact operations that the reductions from \pglcp to \plcp and \swsHS to \swsTwoHS perform. It would be very interesting to analyze whether these reductions actually perform the same operation as the \GridUSO to \CubeUSO reduction (\Cref{thm:generalizationIsUSO}), i.e., whether these reductions commute. It would also be interesting to study whether the \GridUSO to \CubeUSO preserves realizability, i.e., whether if there exists a \pglcp instance inducing a certain grid USO, there also exists a \plcp instance inducing the resulting cube USO.

\newpage
\bibliography{papers,USO}

\end{document}